\documentclass[smallcondensed,draft,fleqn]{svjour3}
\usepackage{amsfonts}
\usepackage{amsmath}
\usepackage{amssymb}
\let\origtheorem\theorem
\let\origproof\proof
\usepackage{z-eves}
\let\theorem\origtheorem
\let\proof\origproof
\usepackage[utf8]{inputenc}
\usepackage{url}
\usepackage{comment}
\usepackage{soul}
\usepackage{xspace}
\usepackage{xcolor}
\usepackage{booktabs}
\usepackage{colortbl}
\usepackage{tabularx}
\usepackage{comment}
\usepackage{algorithm}
\usepackage{algpseudocode}

\newtheorem{teorema}{Theorem}

\newcommand{\lfun}{\longrightarrow}
\newcommand{\disj}{\parallel}

\newcommand{\set}[2]{\{#1 \sqcup #2\}}
\newcommand{\true}{true}
\newcommand{\false}{false}
\newcommand{\Forall}{foreach}
\newcommand{\Exists}{exists}

\newcommand{\setlog}{$\{log\}$\xspace}

\newcommand{\SATRIS}{SAT_\mathcal{RIS}}
\newcommand{\SATX}{SAT_\mathcal{X}}
\newcommand{\LRIS}{\mathcal{L}_\mathcal{RIS}}
\newcommand{\RQ}{\mathcal{RQ}}
\newcommand{\LRQ}{\mathcal{L}_\RQ}
\newcommand{\SATRQ}{SAT_\RQ}
\newcommand{\Ur}{\mathcal{X}}
\newcommand{\LX}{\mathcal{L}_\Ur}

\renewcommand{\plus}{\mathbin{\scriptstyle\sqcup}}
\newcommand{\ww}{\{\cdot \plus \cdot\}}
\newcommand{\e}{\varnothing}

\newcommand{\Var}{\mathcal{V}}
\newcommand{\FUr}{\mathcal{F}_\Ur}
\newcommand{\sUr}{\mathsf{X}}
\newcommand{\Set}{\mathcal{S}}
\newcommand{\isx}{\mathit{isX}}
\newcommand{\sSet}{\mathsf{Set}}
\newcommand{\TRQ}{\mathcal{T}_\RQ}
\newcommand{\TX}{\mathcal{T}_\Ur}
\newcommand{\FX}{\Phi_\Ur}
\newcommand{\FRQ}{\Phi_\RQ}
\newcommand{\FUE}{\FRQ}
\newcommand{\FU}{\Phi_{\forall}}
\newcommand{\FE}{\Phi_{\exists}}
\newcommand{\FEU}{\Phi_{\exists\hspace{-2pt}\forall}}
\newcommand{\CRQ}{\mathcal{C}_\RQ}
\newcommand{\iS}{\mathcal{R}}
\newcommand{\iD}{\mathcal{D}}
\newcommand{\iF}[1]{(#1)^\iS}

\newcommand{\flt}{\phi}   
\newcommand{\ct}{c}       

\newcommand{\vv}{\vec{v}}

\newcommand{\Fpv}{\phi}
\newcommand{\F}{\Fpv}

\newcommand{\why}[1]{\tag*{{\footnotesize [by #1]}}}

\renewcommand{\iff}{\Leftrightarrow}

\newcommand{\q}{\text{\normalfont\'{}}}
\newcommand{\ql}{\hspace{5pt}\text{\normalfont\'{}}}
\newcommand{\qr}{\text{\normalfont\'{}}\hspace{5pt}}

\definecolor{formula}{gray}{0.9}

\title{A Set-Theoretic Decision Procedure for Quantifier-Free, Decidable
Languages Extended with Restricted Quantifiers}
\titlerunning{A Decision Procedure for Restricted Quantifiers}

\author{Maximiliano Cristi\'a \and Gianfranco Rossi}

\institute{M. Cristi\'a --- corresponding author \at Universidad Nacional de Rosario and CIFASIS,
Rosario, Argentina -- \email{cristia@cifasis-conicet.gov.ar} \and G. Rossi \at
Universit\`a di Parma, Parma, Italy -- \email{gianfranco.rossi@unipr.it}}

\smartqed

\begin{document}
\catcode`\@=\active
\def@{\mathbin{\bullet}}

\newcommand{\risnopattern}[3]{\{ #1 : #2 | #3\}}
\newcommand{\risnofilter}[3]{\{ #1 : #2 @ #3\}}
\newcommand{\risnocp}[2]{\{ #1 | #2\}}
\newcommand{\riss}[3]{\{ #1 | #2 @ #3\}}

\newcommand{\defris}[1]{\riss{#1}{\Fpv}{\Ppv}}
\newcommand{\emptyris}[0]{\defris{\e}}
\newcommand{\defrisvar}[0]{\defris{\bar{D}}}
\newcommand{\defriss}[0]{\defris{D}}
\newcommand{\defrisinit}[0]{\defris{\set{d}{D}}}

\maketitle

\begin{abstract}
Let $\LX$ be the language of  first-order, decidable theory $\Ur$. Consider the
language, $\LRQ(\mathcal{X})$, that extends $\LX$ with formulas of the form
$\forall x \in A: \phi$ (restricted universal quantifier, RUQ) and $\exists x
\in A: \phi$ (restricted existential quantifier, REQ), where $A$ is a finite
set and $\phi$ is a formula made of $\Ur$-formulas, RUQ and REQ. That is,
$\LRQ(\mathcal{X})$ admits nested restricted quantifiers. In this paper we
present a decision procedure for $\LRQ(\mathcal{X})$ based on the decision
procedure already defined for the Boolean algebra of finite sets extended with
restricted intensional sets ($\LRIS$). The implementation of the decision
procedure as part of the \setlog (`setlog') tool is also introduced. The usefulness of the
approach is shown through a number of examples drawn from several real-world
case studies.
\end{abstract}

\section{Introduction}\label{sec:intro}

Restricted quantifiers (RQ) are formulas of the following forms:
\begin{gather}
\forall x \in A: \phi(x) \\
\exists x \in A: \phi(x)
\end{gather}
where $A$ is a set called \emph{quantification domain}. The first form is
called \emph{restricted universal quantifier} (RUQ), while the second is called
\emph{restricted existential quantifier} (REQ). The semantics of such formulas
is, respectively:
\begin{gather}
\forall x (x \in A \implies \phi(x)) \\
\exists x (x \in A \land \phi(x))
\end{gather}

RQ are present in formal notations such as B \cite{schneider2001b}, TLA+
\cite{DBLP:books/aw/Lamport2002} and Z \cite{Woodcock00} making it important to
be able to automatically reason about RQ. In fact, RQ allow to express
important program or system properties. For example, one may need to express
that some property, $\phi$, holds for all the users ($U$) of a system. Then, it
can be expressed by means of a RUQ:
\begin{equation*}
I \defs \forall u \in U: \phi(u)
\end{equation*}
Later, one might need to prove that $I$ is a state invariant of
that system by discharging proofs of the form:
\begin{equation}\label{eq:inv}
I \land T \implies I'
\end{equation}
where $T$ is a state transition and $I'$ is the result of substituting every
state variable $v$ in $I$ by $v'$---i.e., the next-state variable. In this
scenario it would be important if many (or all) of those proofs can be
performed automatically.

In a recent article \cite{DBLP:journals/jar/CristiaR21a} we have presented a
decision procedure for a language based on extensional and intensional sets
called $\LRIS(\Ur)$, where $\Ur$ is a first-order, decidable theory. $\LRIS$
can express RQ where the inner formula does not contain other RQ. Then, $\LRIS$
can automatically discharge a proof such as \eqref{eq:inv}. However, in
general, $\LRIS$ does not allow for nested RQ. For example, if $\psi$ is a
formula depending on a user and a process, the following:
\begin{equation}\label{eq:nested}
I_2 \defs \forall u \in U: (\forall p \in P : \psi(u,p))
\end{equation}
where $P$ is the set of processes of the system, is a $\LRIS$ formula only if
$U$ is not part of $\psi$. Consequently, the decision procedure defined for
$\LRIS$ is unable to automatically reason about all formulas such as
\eqref{eq:inv} where $I$ is substituted by $I_2$. Therefore, finding a decision
procedure for formulas of the form \eqref{eq:inv} but involving predicates such
as $I_2$ would be a valuable contribution to the formal verification community.

In this paper we depart from $\LRIS(\Ur)$ to define a new language,
$\LRQ(\Ur)$, admitting finitely nested RQ at any level, when $\LX$ is a
quantifier-free, first-order, decidable language\footnote{Although in
$\LRIS(\Ur)$, $\LX$ can be a quantified language it make little sense to extend
such a language with RQ. Hence, in this paper we focus on quantifier-free
languages that need to be extended to support at least a restricted form of
quantification.}. Then, departing from the decision procedure defined for
$\LRIS$ we define a decision procedure for $\LRQ$. In particular, we provide a
precise condition defining the class of decidable formulas where RUQ and REQ
can be arbitrarily nested. The implementation of these results as part of the
\setlog (`setlog') tool \cite{setlog} is also briefly discussed. More space is committed
to show that the implementation works in practice by providing several examples
of non trivial properties and verification conditions, drawn from real-world
case studies, that \setlog is able to deal with.

The paper is structured as follows. Section \ref{lrq} introduces $\LRQ$ by
first giving an informal presentation (\ref{informal}) and then its formal
syntax (\ref{syntax}) and semantics (\ref{semantics}). The solver for $\LRQ$ is
presented in Section \ref{solver}. In Section \ref{deci} soundness and
completeness (\ref{sandc}) and termination (\ref{term}) of the solver are
proved. Some extensions to $\LRQ$, helping to avoid the introduction of
existential variables, are introduced in Section \ref{extensions}. The
implementation of $\LRQ$ and its solver as part of \setlog is shown in Section
\ref{practice}. In that section we also comment on three case studies carried
out with \setlog involving RQ. Our results are discussed and compared with
similar works in Section \ref{relwork}. Section \ref{concl} gives our
conclusions.

\section{\label{lrq}Formal Syntax and Semantics}

This section describes the syntax and semantics of the set-theoretic language
of Restricted Quantifiers, $\LRQ$. In other words, $\LRQ$ builds RQ from
fundamental concepts drawn from set theory. A gentle, informal introduction is
provided in Section \ref{informal}, followed by the formal presentation of the
language.

$\LRQ$ is a first-order predicate language with terms of sort set and terms
designating ur-elements\footnote{Ur-elements (also known as atoms or
individuals) are objects which contain no elements but are distinct from the
empty set.}. The latter are provided by an external first-order theory $\Ur$
(i.e., $\LRQ$ is parametric with respect to $\Ur$). $\Ur$ must include: a class
$\Phi_\Ur$ of admissible $\Ur$-formulas based on a set of function symbols
$\FUr$ and a set of predicate symbols $\Pi_\Ur$ (providing at least equality);
an interpretation structure $\mathcal{I}_\Ur$ with domain $\iD_\sUr$ and
interpretation function ${(\cdot)}^{\mathcal{I}_\Ur}$; and a decision procedure
$\SATX$ for $\Ur$-formulas. $\LRQ(\Ur)$ denotes the instance of $\LRQ$ based on
theory $\Ur$.

$\LRQ$ provides special set constructors, and a handful of basic predicate
symbols endowed with a pre-designated set-theoretic meaning. Set constructors
are used to construct both restricted intensional sets (RIS) and extensional
sets. Set elements are the objects provided by $\Ur$, which are manipulated
through the primitive operators that $\Ur$ offers. Hence, $\LRQ$ sets represent
\emph{untyped unbounded finite hybrid sets}, i.e., unbounded finite sets whose
elements are of arbitrary sorts. $\LRQ$ formulas are built in the usual way by
using conjunction and disjunction of atomic formulas.

\subsection{$\LRQ$ in a Nutshell}\label{informal}

$\LRQ$ provides three kinds of set terms: $\e$, the empty set; $\{x \plus
A\}$, called \emph{extensional set} whose interpretation is $\{x\} \cup A$; and
$\{\ct:D | \phi\}$, called \emph{restricted intensional set} (RIS) whose
interpretation is $ \{\ct : \ct \in D \land \flt(c)\} $, where $D$ is called
\emph{domain} and $\flt$ is called \emph{filter}. At the same time, $\LRQ$ is a
parametric language w.r.t. the language of some theory $\Ur$. The elements of
sets are $\Ur$ elements and the filters of RIS can be either $\Ur$ formulas or
a very specific kind of $\LRQ$ formulas. $\Ur$ is expected to be a decidable
theory providing at least equality. For example, if $\Ur$ is the theory of
linear integer arithmetic (LIA) then $\LRQ(\Ur)$ will allow to reason about
formulas combining RQ over integer formulas.

In $\LRQ$ formulas are conjunctions and disjunctions of $\LRQ$ and $\Ur$
constraints. In turn, $\LRQ$ provides the set equality ($=$), membership
($\in$) and subset ($\subseteq$) relations, as  constraints.

\begin{example}\label{ex:min}
If $\Ur$ is the theory of LIA then the following is a $\LRQ(\Ur)$ formula:
\[
min \in \{y \plus S\} \land \{y \plus S\} \subseteq \{x:\{y \plus S\} | min \leq x\}
\]
where $min \leq x$ is a $\Ur$ constraint.
\qed
\end{example}

$\LRQ$ allows the definition of RUQ in set-theoretic terms by exploiting the
following identity:
\begin{equation}\label{eq:ruqid}
\forall x \in A: \phi(x) \iff A \subseteq \{x: x \in A \land \phi(x)\}
\end{equation}
In this way, in $\LRQ$ we can define a constraint for RUQ as follows:
\begin{equation}
\Forall(x \in A, \phi) \defs A \subseteq \{x:A | \phi\}
\end{equation}
Then, the formula of Example \ref{ex:min} can be written more compactly as follows:
\[
min \in \{y \plus S\} \land \Forall(x \in \{y \plus S\}, min \leq x)
\]
Likewise, REQ can also be defined as constraints:
\[
\Exists(x \in A, \phi) \defs n \in A \land \phi(n)
\]
where $n$ is a new variable.

Furthermore, in $\LRQ$, the filter of a RIS can be a conjunction of $\Ur$,
$\Forall$ and $\Exists$ constraints. This is an important difference w.r.t.
$\LRIS$ \cite{DBLP:journals/jar/CristiaR21a} because there, RIS filters can
only be $\Ur$ formulas. The possibility of including $\Forall$ and $\Exists$
constraints in RIS filters allows for the definition of nested RQ (what is not
possible in $\LRIS$). For example:
\begin{gather*}
\forall x \in X: (\forall y \in Y: \phi(x,y))) \iff X
  \subseteq \{x: x \in X \land (\forall y \in Y: \phi(x,y))\} \\
\iff X \subseteq \{x: x \in X \land (Y \subseteq \{y: (y \in Y: \phi(x,y))\})\}
\end{gather*}
The latter being equivalent to the following $\LRQ$ formula:
\begin{equation*}
\Forall(x \in X, \Forall(y \in Y, \phi(x,y)))
\end{equation*}
which can be further simplified by introducing some syntactic sugar:
\begin{equation*}
\Forall([x \in X,y \in Y], \phi(x,y))
\end{equation*}

In the next two subsections a formal presentation of $\LRQ$ is made and in
later sections its decidability is analyzed.

\subsection{Syntax}\label{syntax}

The $\LRQ$ syntax is defined primarily by giving the signature upon which terms and
formulas of the language are built.

\begin{definition}[Signature]\label{signature}
The signature $\Sigma_\RQ$ of $\LRQ$ is a triple $\langle
\mathcal{F},\Pi,\Var\rangle$ where:
\begin{itemize}
\item $\mathcal{F}$ is the set of function symbols, partitioned as
$\mathcal{F} = \mathcal{F}_\mathcal{S} \cup \FUr$, where
$\mathcal{F}_\mathcal{S}$ contains $\e$, $\ww$ and $\{\cdot:\cdot | \cdot\}$,
while $\FUr$ contains the function symbols provided by the theory $\Ur$ (at
least, a constant and the binary function symbol $(\cdot,\cdot)$).

\item $\Pi$ is the set of \emph{primitive} predicate symbols, partitioned as
$\Pi = \Pi_\Set \cup \Pi_\mathcal{T} \cup \Pi_\Ur$ where $\Pi_\Set \defs
\{=_\Set,\in,\subseteq\}$ and $\Pi_\mathcal{T}
\defs \{set, \isx\}$, while $\Pi_\Ur$
contains the predicate symbols provided by the theory $\Ur$ (at least $=_\Ur$).
%

\item $\Var$ is a denumerable set of variables, partitioned as
$\Var = \Var_\Set \cup \Var_\Ur$. \qed
\end{itemize}
\end{definition}

Intuitively, $\e$, $\ww$ and $\{\cdot:\cdot | \cdot\}$ are
interpreted as outlined at the beginning of Section \ref{informal}. $=_\Ur$ is
interpreted as the identity in $\iD_\sUr$, while $(\cdot,\cdot)$ will be used
to represent ordered pairs.

Sorts of function and predicate symbols are specified as follows: if $f$
(resp., $\pi$) is a function (resp., a predicate) symbol of arity $n$, then its
sort is an $n+1$-tuple $\langle s_1, \ldots ,s_{n+1} \rangle$ (resp., an
$n$-tuple $\langle s_1, \ldots ,s_n \rangle$) of non-empty subsets of the set
$\{ \sSet , \sUr \}$ of sorts. This notion is denoted by $f:\langle s_1, \ldots
,s_{n+1}\rangle$ (resp., by $\pi:\langle s_1, \ldots ,s_n\rangle $).
Specifically, the sorts of the elements of $\mathcal{F}$ and $\Var$ are the
following. 

\begin{definition}[Sorts of function symbols and variables]\label{d:sorts}
The sorts of the symbols in $\mathcal{F}$ are as follows:
\begin{gather*}
\e: \langle \{\sSet \} \rangle \\
\mathsf{\ww: \langle \{\sUr\},\{ \sSet \} , \{ \sSet\}\rangle } \\
\mathsf{\{\cdot:\cdot | \cdot\}: \langle \{\sUr\}, \{ \sSet \},} \{ \FRQ \}, \mathsf{\{ \sSet\}\rangle } \\
f: \langle \underbrace{\{\sUr\}, \ldots,\{\sUr\}}_n, \{{\sf
\sUr}\}\rangle\text{, if $f \in \FUr$ is of arity $n \ge 0$}.
\end{gather*}
where $\FRQ$ represents the set of $\RQ$-formulas
defined in Definition \ref{formula}. The sorts of variables are as follows:
\begin{gather*}
v: \langle \{\sSet \} \rangle \text{, if $v \in \Var_\Set$} \\
v: \langle \{\sUr \} \rangle \text{, if $v \in \Var_\Ur$}  \tag*{\qed}
\end{gather*}
\end{definition}

\begin{definition}[Sorts of predicate symbols]\label{d:sorts_pred}
The sorts of the predicate symbols in $\Pi$ are as follows:
\begin{gather*}
=_\Set: \langle \{\sSet\} , \{\sSet\} \rangle \\
=_\Ur: \langle \{\sUr\} , \{\sUr\} \rangle \\
\in: \langle \{\sUr \} , \{\sSet \} \rangle \\
\subseteq: \langle \{\sSet \} , \{\sSet \} \rangle \\
set, \isx: \langle \{\sSet,\sUr\} \rangle \tag*{\qed}
\end{gather*}
\end{definition}

Whenever it is clear from context we will write $=$ instead of $=_\Set$ or $=_\Ur$.

\begin{definition}[$\RQ$-terms]\label{RQ-terms}
Let $\TRQ^0$ be the set of terms generated by the
following grammar:
\begin{gather*}
\TRQ^0 ::= \mathit{Elem} \quad|\quad \mathit{Set} \\
\mathit{Elem}::= \TX \quad|\quad \Var_\Ur \\
\mathit{Ctrl} ::=
  \Var_\Ur
  \hspace{2pt}| \ql(\q~~\mathit{Ctrl}~~\ql,\q~~\mathit{Ctrl} \ql)\q \\
Ext ::= \q\e\q \quad|\quad \Var_\Set \quad|\quad \q\{\qr Elem
           \ql\hspace{-2pt}\plus\hspace{-2pt}\qr Ext \ql\}\q \\
\mathit{Ris} ::=
  \ql\{\qr \mathit{Ctrl} \ql:\qr \mathit{Ext}
  \ql\hspace{-2pt}|\hspace{-2pt}\qr \FUE \ql\}\qr \\
Set ::= Ris \quad|\quad Ext
\end{gather*}
where $\TX$ represents the set of non-variable $\Ur$-terms; $\FUE$ is the set
of $\RQ$-formulas defined in Definition \ref{formula}; and variables occurring in a $Ctrl$-term must all be distinct
from each other.

The set of \emph{$\RQ$-terms}, denoted by $\TRQ$, is the maximal
subset of $\TRQ^0$ complying with the sorts as given in Definition
\ref{d:sorts}.
 \qed
\end{definition}

If $t$ is a term $f(t_1,\dots,t_n)$, $f \in \mathcal{F}, n \ge 0$, and $\langle
s_1, \ldots ,s_{n+1} \rangle$ is the sort of $f$, then we say that $t$ is of
sort $\langle s_{n+1} \rangle$. The sort of any $\RQ$-term $t$ is always
$\langle \{\sSet\}\rangle$ or $\langle \{\sUr\} \rangle$. For the sake of
simplicity, we simply say that $t$ is of sort $\sSet$ or $\sUr$, respectively.
In particular, we say that a $\RQ$-term of sort $\sSet$ is a \emph{set term},
that set terms of the form  $\{t_1 \plus t_2\}$ are \emph{extensional} set
terms, and that terms of the form $\{t_1:t_2 | \phi\}$ are RIS terms. The first
argument of an extensional set term is called \emph{element part} and the
second is called \emph{set part}. In turn, the first argument of a RIS term is
called \emph{control term}, the second is the \emph{domain} and the third one
is the \emph{filter}. 

As can be seen in Definition \ref{RQ-terms}, control terms can be either
variables or nested ordered pairs. The utility of the latter will be precisely
motivated and discussed in Section \ref{extensions}. Note that the domain of a
RIS term can be the empty set, a set variable or an extensional set.

Hereafter, we will use the following notation for extensional set terms:
$\{t_1,t_2,\dots,t_n \plus t\}$, $n \ge 1$, is a shorthand for $\{t_1
\plus \{t_2 \,\plus\, \cdots \{ t_n \plus t\}\cdots\}\}$, while
$\{t_1,t_2,\dots,t_n\}$ is a shorthand for $\{t_1,t_2,\dots,t_n \plus \e\}$.

\begin{definition}[$\RQ$-constraints]\label{primitive-constraint}
If $\pi \in \Pi$ is a predicate symbol of sort $\langle s_1, \ldots ,
s_n \rangle$, and for each $i=1,\ldots , n$, $t_i$ is a $\RQ$-term of sort
$\langle s'_i \rangle$ with $s'_i \subseteq s_i$, then:
\begin{enumerate}
\item If $\pi$  is $\subseteq$, then $\pi(t_1,t_2)$ is a $\RQ$-constraint
if $t_2 \equiv \{Ctrl:t_1 | \FUE\}$, where $Ctrl$ and $\FUE$ are as in
Definition \ref{RQ-terms}.
\item If $\pi$ is $\in$, then $\pi(t_1,t_2)$ is a $\RQ$-constraint
if $t_2$ is an $Ext$ term as in Definition \ref{RQ-terms}.
\item If $\pi$ is $=_\Set$, then $\pi(t_1,t_2)$ is a $\RQ$-constraint
if $t_1$ and $t_2$ are $Ext$ terms as in Definition \ref{RQ-terms}.
\item If $\pi$ is any other element of $\Pi$,
then $\pi (t_1,\ldots ,t_n)$ a $\RQ$-constraint.
\end{enumerate}
The set of $\RQ$-constraints is denoted by $\CRQ$.
\qed
\end{definition}

The $\RQ$-constraints based on symbols in $\Pi_\Set$ will be called \emph{set
constraints}. Note that the conditions on $\subseteq$-constraints forces them
to be RUQ as in \eqref{eq:ruqid}.

Finally, we define the set of $\RQ$-formulas as follows.

\begin{definition}[$\RQ$-formulas]\label{formula}
The set of
$\RQ$-formulas, denoted by $\FRQ$, is given by the following grammar:
\begin{gather*}
\FRQ ::=
  \true
  \mid \false
  \mid \FX
  \mid \CRQ
  \mid \FRQ \land \FRQ
  \mid \FRQ \lor \FRQ
\end{gather*}
where $\FX$ and $\CRQ$ represent any element belonging to the class of
$\Ur$-formulas and $\RQ$-constraints, respectively.
\qed
\end{definition}

As can be seen, $\LRQ$ is based solely on fundamental concepts of set theory.

\begin{remark}[Notation]
We will use the following naming conventions, unless stated
differently: $A, B, C, D$
stand for terms of sort $\sSet$;  $a, b, c, d, x, y, z$ stand for terms of sort
$\sUr$; and $t, u, v$ stand for terms of any of the two sorts. A symbol such as
$\dot{A}$ states that $A \in \Var$. Finally, $n, n_i$ stand for new variables
of sort $\sUr$; and $N, N_i$ for new variables of sort $\sSet$; no dot above
them will be used. \qed
\end{remark}

\begin{remark}[$\LRQ$ vs. $\LRIS$]\label{r:lris}
As we have pointed out in Section \ref{sec:intro}, $\LRQ$ departs from $\LRIS$
\cite{DBLP:journals/jar/CristiaR21a}. $\LRQ$ is a sublanguage of $\LRIS$ except
for one modification which extends $\LRIS$. Indeed, $\LRIS$ admits the same
function and predicate symbols than $\LRQ$, plus some other or more complex
versions of them. For example, in $\LRIS$ RIS terms have a more complex
structure and union, intersection, etc. of RIS terms are available.

However, in $\LRIS$ filters can only be $\LX$ formulas. In $\LRQ$ filters can
be nested RQ ending in an $\LX$ formula. This extension
is crucial to extend the expressiveness of the language (cf. formula
\eqref{eq:nested}). The restriction on filters to $\LX$ formulas in $\LRIS$ is
key to define a decision procedure for it. If this restriction is lifted,
termination of the decision algorithm is compromised. As we will shown in
Section \ref{deci}, there are subclasses of formulas admitting nested RQ that
do not compromise termination. \qed
\end{remark}

\subsection{\label{semantics}Semantics}

Sorts and symbols in $\Sigma_{\RQ}$ are interpreted according to the
interpretation structure $\iS \defs \langle \iD,\iF{\cdot}\rangle$, where $\iD$ and
$\iF{\cdot}$ are defined as follows.

\begin{definition} [Interpretation domain] \label{def:int_dom}
The interpretation domain $D$ is partitioned as $\iD \defs \iD_\sSet \cup \iD_\sUr$
where:
\begin{itemize}
\item $\iD_\sSet$ is the set of all hereditarily finite hybrid
sets built from elements in $\iD$. Hereditarily finite sets are those sets that
admit (hereditarily finite) sets as their elements, that is sets of sets.
\item $\iD_\sUr$ is a collection of other objects.
\qed
\end{itemize}
\end{definition}

\begin{definition} [Interpretation function] \label{app:def:int_funct}
The interpretation function $\iF{\cdot}$ is defined as follows:
\begin{itemize}
\item Each sort $\mathsf{S} \in \{\sSet,\sUr\}$ is mapped to
      the domain $\iD_\mathsf{S}$.

\item For each sort $\mathsf{S}$, each variable $x$ of sort $\mathsf{S}$ is mapped to
      an element $x^\iS$ in  $\iD_\mathsf{S}$.

\item The constant and function symbols in $\mathcal{F}_\Set$
are mapped to
elements in $\iD_\mathsf{S}$ as follows:
  \begin{itemize}
  \item $\e$ is interpreted as the empty set, namely $\e^\iS = \emptyset$
  \item $\{ x \plus A \}$ is interpreted as the set $\{x^\iS\} \cup A^\iS$.
  \item Let $\vec{x}$ be a vector of variables occurring in $\ct$
  and $\vv$ a vector of other variables, then the set
  $\{\ct(\vec{x}):X | \flt(\vec{x},\vv)\}$ is interpreted as the set:
\[
\{y : \exists \vec{x} (\ct(\vec{x}) \in_\Ur X \land \flt(\vec{x},\vv))\}
\]
Note that in RIS terms, $\vec{x}$ are ``local'' variables whose scope is the
RIS itself, while $\vv$ are ``non-local'' variables whose scope is the formula
where the RIS is participating in.
  \end{itemize}

\item The predicate symbols in $\Pi$ are interpreted as follows:
  \begin{itemize}
   \item $A =_\mathcal{S} B$ is interpreted as $A^\iS = B^\iS$,
   where $=$ is the identity relation in $\iD_\sSet$
   \item $x =_\Ur y$ is interpreted as $x^\iS = y^\iS$,
   where $=$ is the identity relation in $\iD_\sUr$
   \item $x \in A$ is interpreted as $x^\iS \in A^\iS$
   \item $A \subseteq B$ is interpreted as $A^\iS \subseteq B^\iS$
   \item $\isx(t)$ is interpreted as $t^\iS \in \iD_\sUr$
   \item $set(t)$ is interpreted as $t^\iS \in \iD_\sSet$.
\qed
\end{itemize}
\end{itemize}
\end{definition}

The interpretation structure $\iS$ is used to map each
$\RQ$-formula $\Phi$ to a truth value $\Phi^\iS = \{\true,\false\}$ in
the following way: set constraints (resp., $\Ur$ constraints) are evaluated by
$\iF{\cdot}$ according to the meaning of the corresponding predicates in set
theory (resp., in theory $\Ur$) as defined above; $\RQ$-formulas are evaluated
by $\iF{\cdot}$ according to the rules of propositional logic. A $\LRQ$-formula
$\Phi$ is \emph{satisfiable} iff there exists an assignment $\sigma$ of values
from $\mathcal{D}$ to the free variables of $\Phi$, respecting the sorts
of the variables, such that $\Phi[\sigma]$ is true in $\iS$, i.e., $\iS \models
\Phi[\sigma]$. In this case, we say that $\sigma$ is a \emph{successful
valuation} (or, simply, a \emph{solution}) of $\Phi$.

\section{\label{solver}A Solver for $\LRQ$}

In this section we present a constraint solver for $\LRQ$, called $\SATRQ$. The
solver provides a collection of rewrite rules for rewriting $\LRQ$ formulas
that are proved to be a decision procedure for some subclasses of $\LRQ$
formulas (see Section \ref{deci}). As already observed, however, checking the
satisfiability of $\RQ$-formulas depends on the existence of a decision
procedure for $\Ur$-formulas (i.e., formulas over $\LX$).

\subsection{The Solver}

$\SATRQ$ is a rewriting system whose global organization is shown in Algorithm
\ref{glob}, where $\textsf{STEP}$ is the core of the algorithm.

\textsf{sort\_infer} is used to automatically add $\Pi_\mathcal{T}$-constraints
to the input formula $\Phi$ to force arguments of $\RQ$-constraints
in $\Phi$ to be of the proper sorts (see Remark \ref{rsort} below).
\textsf{sort\_infer} is called twice in Algorithm \ref{glob}: first, at the
beginning of the algorithm, and second, within procedure $\textsf{STEP}$ for
the constraints that are generated during constraint processing.
\textsf{sort\_check} checks $\Pi_\mathcal{T}$-constraints occurring in $\Phi$:
if they are satisfiable, then $\Phi$ is returned unchanged; otherwise, $\Phi$
is rewritten to $\false$.

\algtext*{EndIf}
\algrenewtext{EndProcedure}{\textbf{return} }
\begin{algorithm}
\begin{tabular}{lr}
\begin{minipage}{.55\textwidth}
\begin{algorithmic}[0]
\Procedure{$\mathsf{STEP}$}{$\Phi$}
\State \hspace{5mm}
       \textsf{for all} $\pi \in \Pi_\mathcal{S} \cup \Pi_\mathcal{T}:
            \Phi \gets \mathsf{rw}_{\pi}(\Phi)$;
\State \hspace{5mm}
       $\Phi \gets \mathsf{sort\_check}(\mathsf{sort\_infer}(\Phi))$
\EndProcedure $\Phi$
\end{algorithmic}
\begin{algorithmic}[0]
 \Procedure{$\mathsf{rw}_\pi$}{$\Phi$}
 \If{\text{$\Phi = \mbox{}\dots \land \false \land \cdots$}}
 \State\Return{$\false$}
 \Else
 \Repeat
 \State \textsf{let } $\Phi$ \textsf{ be } $c_1 \land \dots \land c_m$
 \State \textsf{select a $\pi$-constraint $c_i$ in $\Phi$}
 \State \textsf{apply any applicable rule to $c_i$}
 \Until{\textsf{no rule applies to any $\pi$-constraint }}
 \EndIf
 \EndProcedure $\Phi$
\end{algorithmic}
\end{minipage}
&
\begin{minipage}{.35\textwidth}
\begin{algorithmic}[0]
\Procedure{$\SATRQ$}{$\Phi$}
 \State $\Phi \gets \textsf{sort\_infer}(\Phi)$
 \Repeat
   \State $\Phi' \gets \Phi$
   \State $\Phi \gets \textsf{STEP}(\Phi)$
 \Until{$\Phi = \Phi'$}
 \State $\Phi \textbf{ is } \Phi_\mathcal{S} \land \Phi_\Ur$
 \State $\Phi \gets \Phi_\mathcal{S} \land \SATX(\Phi_\Ur)$
\EndProcedure $\Phi$
\end{algorithmic}
\end{minipage}
\end{tabular}
\caption{\label{glob}The $\SATRQ$ solver. $\Phi$ is the input formula.}
\end{algorithm}

\textsf{STEP} applies specialized rewriting procedures to the current formula
$\Phi$ and returns either $\false$ or the modified formula. Each rewriting
procedure applies a few non-deterministic rewrite rules which reduce the
syntactic complexity of $\RQ$-constraints of one kind. Procedure
$\mathsf{rw}_\pi$ in Algorithm \ref{glob} represents the rewriting procedure
for ($\Pi_\mathcal{S} \cup \Pi_\mathcal{T}$)-constraints. The execution of
$\textsf{STEP}$ is iterated until a fixpoint is reached---i.e., the formula
cannot be simplified any further. \textsf{STEP} returns $\false$ whenever (at
least) one of the procedures in it rewrites $\Phi$ to $\false$. In this case, a
fixpoint is immediately detected, since $\textsf{STEP}(\false)$ returns
$\false$.

$\SATX$ is the constraint solver for $\Ur$-formulas. The formula $\Phi$ can be
seen, without loss of generality, as $\Phi_\mathcal{S} \land \Phi_\Ur$, where
$\Phi_\mathcal{S}$ is a pure $\RQ$-formula (basically, a $\RQ$-formula
with with no $\Ur$-formula in it---see Definition \ref{d:subclasses}) and
$\Phi_\Ur$ is an $\Ur$-formula. $\SATX$ is applied only to the $\Phi_\Ur$
conjunct of $\Phi$. Note that, conversely, \textsf{STEP} rewrites only
$\RQ$-constraints, while it leaves all other atoms unchanged. Nonetheless, as
the rewrite rules show, $\SATRQ$ generates $\Ur$-formulas that are conjoined to
$\Phi_\Ur$ so they are later solved by $\SATX$. 

As we will show in Section \ref{deci}, when all the non-deterministic
computations of $\SATRQ(\Phi)$ return $\false$, then we can conclude that
$\Phi$ is unsatisfiable; otherwise, we can conclude that $\Phi$ is satisfiable
and each solution of the formulas returned by $\SATRQ$ is a solution of
$\Phi$, and vice versa.

\begin{remark}\label{rsort}
$\LRQ$ does not provide variable declarations. The sort of variables are
enforced by adding suitable \emph{sort constraints} to the formula to be
processed. Sort constraints are automatically added by the solver.
Specifically, a constraint $set(y)$ (resp., $\isx(y)$) is added for each
variable $y$ which is required to be of sort $set$ (resp., $\sUr$). For
example, given $B = \{y \plus A\}$, \textsf{sort\_infer} conjoins the sort
constraints $set(B)$, $\isx(y)$ and $set(A)$. If the set of function and
predicate symbols of $\LRQ$ and $\LX$ are disjoint, there is a unique
sort constraint for each variable in the formula.
\qed
\end{remark}

\subsection{\label{rules}Rewrite Rules}

The rewrite rules used by $\SATRQ$ are defined as follows.

\begin{definition}[Rewrite rules]\label{d:rw_rules}
If $\pi$ is a symbol in $\Pi_\mathcal{S} \cup \Pi_\mathcal{T}$ and $p$ is a
$\RQ$-constraint based on $\pi$, then a \emph{rewrite rule for
$\pi$-constraints} is a rule of the form $p \lfun \Phi_1 \lor \dots \lor
\Phi_n$, where $\Phi_i$, $i \ge 1$, are $\RQ$-formulas. Each
atom matching $p$ is non-deterministically rewritten to
one of the $\Phi_i$. Variables appearing in the right-hand side but not in the
left-hand side are assumed to be fresh variables, implicitly existentially
quantified over each $\Phi_i$.
Conjunction has higher precedence than disjunction.
\qed
\end{definition}

A \emph{rewriting procedure} for $\pi$-constraints consists of the collection
of all the rewrite rules for $\pi$-constraints. For each rewriting procedure,
\textsf{STEP} selects rules in the order they are listed in Figure
\ref{f:rules}. The first rule whose left-hand side matches the input
$\pi$-constraint is used to rewrite it.

\begin{figure}
\hrule
\begin{gather}
\textsc{Subset} \notag \\
\e \subseteq \risnocp{x:A}{\Fpv(x)}
 \lfun \true \label{forall:empty} \\[1mm]
\set{a}{A} \subseteq \risnocp{x:\set{a}{A}}{\Fpv(x)}
  \lfun
    \Fpv(a)
     \land A \subseteq \risnocp{x:A}{\Fpv(x)}
  \label{forall:iter} \\[1mm]
\dot{A} \subseteq \risnopattern{x}{\dot{A}}{\Fpv(x)}
  \lfun \textsc{irreducible} \label{forall:var} \\[5mm]
\textsc{Membership} \notag \\
a \in \e \lfun \false \\[1mm]
a \in \set{b}{A} \lfun a =_\Ur b \lor a \in A
  \label{eq:inext} \\[1mm]
a \in \dot{A} \lfun \dot{A} = \set{a}{N} \label{eq:invar} \\[5mm]
\textsc{Equality} \notag \\
\e = \e \lfun \true \\[1mm]
\dot{A} = \dot{A} \lfun \true \\[1mm]
B = \dot{A} \lfun \dot{A} = B \text{, if $B \notin \Var$} \\[1mm]
\dot{A} = B \lfun \dot{A} = B \text{, and substitute $A$ by $B$ in the rest of the formula} \label{eq:subs} \\[1mm]
\set{a}{A} = \e \lfun \false \\[1mm]
\e = \set{a}{A} \lfun \false \\[1mm]
\set{a}{A} = \set{b}{B} \lfun \label{eq:eqext} \\
\qquad
  a =_\Ur b \land A = B \notag \\
\qquad
  \lor a =_\Ur b \land \set{a}{A} = B \lor
  a =_\Ur b \land A = \set{b}{B} \lor
  A = \set{b}{N} \land B = \set{a}{N} \notag \\[1mm]
\dot{A} = B \lfun \text{\textsc{irreducible}, if $\dot{A}$ does not occur elsewhere in the formula}
\end{gather}
\hrule \caption{\label{f:rules}Rewrite rules for $\RQ$-constraints }
\end{figure}

Rules whose right-hand side is \textsc{irreducible} indicate that the
constraint at the left-hand side is not rewritten and will remain as it is all
the way to the final answer returned by Algorithm \ref{glob}. In Figure
\ref{f:rules}, we have made explicit equality in $\LX$ by means of $=_\Ur$. All
other instances of $=$ correspond to equality in $\LRQ$ (i.e., set equality).

As shown in Figure
\ref{f:rules}, there are rewriting procedures for $\subseteq$-constraints
(\textsc{Subset}), $\in$-constraints (\textsc{Membership}) and $=$-constraints
(\textsc{Equality}). The \textsc{Membership} and \textsc{Equality} rules deal
only with extensional sets due to the restrictions given in Definition
\ref{primitive-constraint}. Observe that all other constraints generated by the
rules of Figure \ref{f:rules} are $\Ur$-constraints which are dealt with by
$\SATX$. All the rules in the figure are borrowed from the $\LRIS$ solver
\cite{DBLP:journals/jar/CristiaR21a}. This is important because it simplifies
the proof of some important properties of $\SATRQ$ (Section \ref{deci}).

As can be seen, most of the rules are straightforward.
Rule \eqref{eq:eqext} is the main rule of set unification \cite{Dovier2006}.
Set unification is pervasive in other logics developed by the authors
\cite{DBLP:journals/jar/CristiaR20,DBLP:journals/jar/CristiaR21a}. This rule
states when two non-empty, non-variable sets are equal by non-deterministically
and recursively computing four cases. These cases implement the
\emph{Absorption} and \emph{Commutativity on the left} properties of set theory
\cite{Dovier00}. As an example, by applying rule \eqref{eq:eqext} to $\{1\} =
\{1,1\}$ we get: ($1 = 1 \land \e = \{1\}) \lor (1 = 1 \land \{1\} = \{1\})
\lor (1 = 1 \land  \e = \{1,1\}) \lor (\e = \{1 \plus \dot N\} \land  \{1 \plus
\dot N\} = \{1\})$, which turns out to be true (due to the second disjunct).

Rules \eqref{forall:empty}-\eqref{forall:var} process RUQ by implementing
\eqref{eq:ruqid}. Rule \eqref{forall:iter} iterates over all the elements of
the domain of the RIS until it becomes the empty set or a variable. In each
iteration one of the elements of the domain is proved to verify the filter (if
not, the rule fails), and a new iteration is fired. If the domain becomes a
variable the constraint is not processed any more. Note that a constraint such
as $\dot{A} \subseteq \risnopattern{x}{\dot{A}}{\Fpv(x)}$ is trivially
satisfied by substituting $\dot{A}$ by the empty set.

As $\LRQ$, $\SATRQ$ is based solely on fundamental concepts of set theory.

\begin{remark}\label{r:gen}
Observe that when $\subseteq$ are rewritten only the following are generated:
\begin{itemize}
\item $\phi \in \FX$
\item $A \subseteq \{x:A | \phi(x)\}$
\end{itemize}
\qed
\end{remark}

\subsection{\label{irreducible}Irreducible Constraints}

When no rewrite rule is applicable to the current $\RQ$-formula $\Phi$ and $\Phi$
is not $\false$, the main loop of $\SATRQ$ terminates returning $\Phi$ as its
result. This formula can be seen, without loss of generality, as $\Phi_\Set
\land \Phi_\Ur$, where $\Phi_\Ur$ contains all (and only) $\Ur$ constraints
and $\Phi_\Set$ contains all other constraints occurring in $\Phi$.

The following definition precisely characterizes the form of atomic constraints
in $\Phi_\Set$.

\begin{definition}[Irreducible formula]\label{def:solved}
Let $\Phi$ and $\phi$ be $\RQ$-formulas, $A \in \Var_\Set$, $x$ a control term
(thus it is a term of sort $\sUr$) and $t$ a term of sort $\sSet$. A
$\RQ$-constraint $p$ occurring in $\Phi$ is \emph{irreducible} if it has one of
the following forms:
\begin{enumerate}
\item \label{i:icfirst} $A = t$, and neither $t$ nor $\Phi \setminus \{A = t\}$ contain $A$
\item $A \subseteq \risnopattern{x}{A}{\Fpv(x)}$
\end{enumerate}
A $\RQ$-formula $\Phi$ is irreducible if it is $\true$ or if all of its
$\RQ$-constraints are irreducible.
\qed
\end{definition}

$\Phi_\Set$, as returned by $\SATRQ$'s main loop, is an irreducible formula.
This fact can be checked by inspecting the rewrite rules presented in Figure
\ref{f:rules}. This inspection is straightforward as there are no rules
rewriting irreducible constraints and all non-irreducible form constraints are
rewritten by some rule.

It is important to observe that the atomic constraints occurring in $\Phi_\Set$
are indeed quite simple. In particular, all non-variable set terms occurring in
the input formula have been removed, except those occurring as right-hand side
of $=$ constraints.

\section{\label{deci}Decidability of $\LRQ$ Formulas}

In this section we analyze the soundness, completeness and termination
properties of $\SATRQ$ for different subclasses of $\RQ$-formulas.

As we have explained in Remark \ref{r:lris}, $\LRQ$ is a sublanguage of $\LRIS$
except for the fact that $\LRQ$ admits RQ in filters. We also pointed out that
accepting RQ in filters poses termination problems; soundness and completeness
are not affected. Actually, as noted in Section \ref{solver}, all the rules of
Figure \ref{f:rules} are rules borrowed from the $\LRIS$ solver. Hence, we will
briefly analyze soundness and completeness of $\SATRQ$ and will spend more time
analyzing its termination.

As RUQ and REQ play a central role in this work, we provide some syntactic
sugar for them.

\begin{definition}[Restricted Quantifiers]\label{d:rq}
Given a control term $x$, an extensional set term $A$ and a formula $\phi$, a
\emph{restricted universal quantifier} (RUQ), noted $\Forall(x \in A, \phi)$,
is defined as:
\begin{equation}\label{eq:foreach}
\Forall(x \in A, \phi) \defs A \subseteq \{x : A | \phi(x)\}
\end{equation}
Under the same terms, a \emph{restricted existential quantifier} (REQ), noted
$\Exists(x \in A, \phi)$, is defined as:
\begin{equation}\label{eq:exists}
\Exists(x \in A, \phi) \defs n \in A \land \phi(n)
\end{equation}
where all variables occurring in $n$ are fresh variables not occurring
elsewhere in the formula of which the REQ is a part. \qed
\end{definition}

In a RQ: $x$ is called \emph{control term} or \emph{quantified variable}, $A$
is called \emph{domain} and $\phi$ is called \emph{filter} (following the
vocabulary of RIS terms). Note that in both RUQ and REQ, $\phi$ must depend on
$x$.

\begin{definition}[Subclasses of $\FUE$]\label{d:subclasses}
The following are the subclasses of $\FUE$-formulas for which decidability will
be analyzed:
\begin{itemize}
\item
$\Phi_{nrq}$ is the subclass of $\FUE$ whose elements are nested RQ. The
subclasses of formulas to be analyzed will be subclasses of $\Phi_{nrq}$.
\begin{gather*}
\FU^\Ur ::= \Forall(Ctrl \in Ext, \FX)\\
\FE^\Ur ::= \Exists(Ctrl \in Ext, \FX)\\
\Phi_{mix} ::=
  \FU^\Ur
  \mid \FE^\Ur
  \mid \Forall(Ctrl \in Ext, \Phi_{mix})
  \mid \Exists(Ctrl \in Ext, \Phi_{mix})\\
\Phi_{nrq} ::=
  \true
  \mid \false
  \mid \FX
  \mid \Phi_{mix}
  \mid \FRQ \land \FRQ
  \mid \FRQ \lor \FRQ
\end{gather*}
\item $\FU$ is the subclass of $\Phi_{nrq}$ whose elements are built from $\FX$ and nested RUQ.
\begin{gather*}
\FU^p ::= \FU^\Ur \mid \Forall(Ctrl \in Ext, \FU^p) \\
\FU ::=
  \FX
  \mid \FU^p
  \mid \FU \land \FU
  \mid \FU \lor \FU
\end{gather*}
\item $\FE$ is the subclass of $\Phi_{nrq}$ whose elements are built from $\FX$ and nested REQ.
\begin{gather*}
\FE^p ::= \FE^\Ur \mid \Exists(Ctrl \in Ext,\FE^p) \\
\FE ::=
  \FX
  \mid \FE^p
  \mid \FE \land \FE
  \mid \FE \lor \FE
\end{gather*}
\item $\FEU$ is the subclass of $\Phi_{nrq}$ whose elements are built from $\FX$ and
nested RQ where all REQ are \emph{before} all RUQ (if any).
\begin{gather*}
\FEU ::=
  \FX
  \mid \Exists(Ctrl \in Ext,Filter)
  \mid \FEU \land \FEU
  \mid \FEU \lor \FEU \\
Filter ::= \FU^p \mid \FEU^p
\end{gather*}
\item $\FRQ^p$ is the subclass of $\Phi_{nrq}$ whose elements are pure $\RQ$-formulas.
\begin{gather*}
\FRQ^p ::=
  \true
  \mid \false
  \mid \Phi_{mix}
  \mid \FRQ^p \land \FRQ^p
  \mid \FRQ^p \lor \FRQ^p
\end{gather*}
Similar definitions can be given for pure $\FU$, $\FE$ and $\FEU$ formulas.
\qed
\end{itemize}
\end{definition}

\begin{example}
Different classes of formulas.
\begin{itemize}
\item $\Forall(x \in A, \Forall(y \in B, x = y))$ is a $\FU$ formula.
\item $\Exists(x \in A, \Exists(y \in B, \lnot x = y))$ is a $\FE$ formula.
\item $\Exists(x \in A, \Forall(y \in B, x = y))$ is a $\FEU$ formula.
\item $\Forall(x \in A, \Exists(y \in B, x = y))$
is a $\FRQ^p$ formula. \qed
\end{itemize}
\end{example}

As can be seen, formulas in $\Phi_{nrq}$ are conjunctions and disjunctions of
$\Ur$-formulas and nested RQ; the filter of the innermost RQ is an
$\Ur$-formula. Note that not every $\RQ$-constraint can be part of a formula in
$\Phi_{nrq}$; the idea is to restrict them to be RQ. Hence, basically, we
analyze the decidability of $\RQ$-formulas strictly encoding RQ. However, note
that when $\SATRQ$ processes a $\Phi_{nrq}$-formula it may generate a formula
outside $\Phi_{nrq}$. For example, $\Exists(x \in \dot{A}, \phi(x))$ is
rewritten into $n \in \dot{A} \land \phi(n)$ which then is rewritten into $A =
\{n \plus N\} \land \phi(n)$, which is not a $\Phi_{nrq}$-formula due to the
presence of $A = \{n \plus N\}$.

\begin{remark}[Notation]
From now on, we will write $\Forall([x \in A,y \in B], \phi)$ as a shorthand
for $\Forall(x \in A, \Forall(y \in B, \phi))$, and $\Exists([x \in
A,y \in B], \phi)$ as a shorthand for $\Exists(x \in A, \Exists(y \in B,
\phi))$. Besides, $\Forall(\vec{x}_n \in \vec{A}_n, \phi)$ denotes
$\Forall([x_1 \in A_1, \dots x_n \in A_n], \phi)$ and $\Exists(\vec{x}_n \in
\vec{A}_n, \phi)$ denotes $\Exists([x_1 \in A_1, \dots x_n \in A_n], \phi)$, $0
\leq n$ (if $n = 0$ we take $\phi$ as the resulting formula).
\qed
\end{remark}

With this notation it is easy to see that: in pure $\FU^p$ formulas there are
only constraints of the form $\Forall(\vec{x}_n \in \vec{A}_n, \phi)$ for some
$\Ur$-formula $\phi$; in pure $\FE^p$ formulas there are only constraints of
the form $\Exists(\vec{x}_n \in \vec{A}_n, \phi)$ for some $\Ur$-formula
$\phi$; and in pure $\FEU^p$ formulas there are only constraints of the form
$\Exists(\vec{x}_n \in \vec{A}_n,\Forall(\vec{y}_m \in \vec{B}_m, \phi))$ for
some $\Ur$-formula $\phi$ ($0 < n, 0 \leq m$).

\subsection{\label{sandc}Soundness and Completeness}
The following theorem ensures that, after termination, the
rewriting process implemented by $\SATRQ$ preserves the set of solutions of
the input formula.

\begin{teorema}[Equisatisfiability]\label{equisatisfiable}
Let $\Phi$ be a $\Phi_{nrq}$-formula and $\Phi^1, \Phi^2,\dots,\Phi^n$ be the
collection of $\RQ$-formulas returned by $\SATRQ(\Phi)$. Then $\Phi^1 \lor
\Phi^2 \lor \dots \lor \Phi^n$ is equisatisfiable to $\Phi$, that is, every
possible solution\footnote{More precisely, each solution of $\Phi$ expanded to
the variables occurring in $\Phi^i$ but not in $\Phi$, so as to account for the
possible fresh variables introduced into $\Phi^i$.} of $\Phi$ is a solution of
one of the $\Phi^i$s and, vice versa, every solution of one of these formulas
is a solution for $\Phi$.
\end{teorema}

\begin{proof}
The proof rests on a series of lemmas each showing that the set of solutions of
left and right-hand sides of each rewrite rule are the same. Given that the
rewrite rules of Figure \ref{f:rules} are those used to define the solver for
$\LRIS$, then the lemmas proved for $\LRIS$ still apply \cite[Appendix
C.4]{DBLP:journals/jar/CristiaR21a}. The only concern with those lemmas might
be the fact that they were proved under the assumption that RIS filters do not
admit RQ. However, it is trivial to see that all the \textsc{Membership} and
\textsc{Equality} rules and rules
\eqref{forall:empty} and \eqref{forall:var} are unaffected by the fact that
filters admit RQ. For the remaining rule, i.e.
\eqref{forall:iter}, we reproduce in Appendix \ref{app:proofs} the proof made
for $\LRIS$ so readers can check that it do not depend on any limitation over
RIS filters. \qed
\end{proof}

\begin{teorema}[Satisfiability of the output formula]\label{satisf}
Any $\RQ$-formula different from $\false$ returned by $\SATRQ$ is
satisfiable w.r.t. the underlying interpretation structure $\iS$.
\end{teorema}

\begin{proof}
As we have explained, each disjunct of the formula returned by $\SATRQ$ can be
written as $\Phi_\mathcal{S} \land \Phi_\Ur$, where $\Phi_\mathcal{S}$ is a
pure $\RQ$-formula and $\Phi_\Ur$
is an $\Ur$-formula.

Since $\SATX$ is called on $\Phi_\Ur$ we know that it is satisfiable (under the
assumption that $\SATRQ$ has not returned $\false$).

Now we prove that $\Phi_\mathcal{S}$ is satisfiable, too. We know that
$\Phi_\mathcal{S}$ is an irreducible formula (Definition \ref{def:solved}).
Then, we have to prove that an irreducible formula is always satisfiable. Given
that an irreducible formula is a conjunction of irreducible constraints, we
have to prove that all these constraints can be simultaneously satisfied.
Constraints of the form $\dot{A} = t$  are satisfied by binding $\dot{A}$ to
$t$ (recall from Definition \ref{def:solved} that $\dot{A}$ does not occur 3in
the rest of an irreducible formula); constraints of the form
$\dot{A} \subseteq \risnopattern{x}{\dot{A}}{\Fpv(x)}$ are satisfied by
substituting the domain of the RIS by the empty set. Hence, there is always a
solution for an irreducible $\RQ$-formula.

Finally, we prove that $\Phi_\mathcal{S} \land \Phi_\Ur$ can be satisfied.
Indeed, observe that the solution for $\Phi_\mathcal{S}$ do not bind variables
of sort $\sUr$ and that $\Phi_\Ur$ do not contain variables of sort $\sSet$. So
the values of the solution for $\Phi_\Ur$ do no conflict with the values of the
solution for  $\Phi_\mathcal{S}$. \qed
\end{proof}

The following example shows how Theorem \ref{satisf} works in practice.

\begin{example}
Consider the following nested RUQ where $\phi$ is an $\Ur$-formula:
\begin{equation}\label{eq:exruq2}
\Forall([x \in \{a \plus \dot{A}\}, y \in \{b \plus \dot{B}\}],\phi(x,y))
\end{equation}
$\SATRQ$ applies \eqref{eq:foreach} and rule \eqref{forall:iter} twice
yielding:
\begin{equation*}
\phi(a,b) \land \Forall(y \in \dot{B}, \phi(a,y)) \land \Forall([x \in \dot{A}, y \in \{b \plus \dot{B}\}],\phi(x,y))
\end{equation*}
Now, it calls $\SATX(\phi(a,b))$ because both $\Forall$ constraints are
irreducible. Thus, determining the satisfiability of \eqref{eq:exruq2} is
reduced to determining the satisfiability of $\phi(a,b)$ because satisfiability
of the two $\Forall$ constraints is guaranteed by Theorem \ref{satisf} (with $A
\leftarrow \e, B \leftarrow \e$). \qed
\end{example}

Thanks to Theorems \ref{equisatisfiable} and \ref{satisf} we can conclude that,
given a $\Phi_{nrq}$-formula $\Phi$, then $\Phi$ is satisfiable with respect to
the intended interpretation structure $\iS$ if and only if there is a
non-deterministic choice in $\SATRQ(\Phi)$ that returns a $\RQ$-formula
different from $\false$. Conversely, if all the non-deterministic computations
of $\SATRQ(\Phi)$ terminate with $\false$, then $\Phi$ is surely unsatisfiable.
Note that these theorems have been proved for any $\Phi_{nrq}$-formula.

\subsection{\label{term}Termination}
The problem is that termination of $\SATRQ$ cannot be proved for every
$\Phi_{nrq}$-formula, as shown by the following example.

\begin{example}\label{ex:undec}
The following nested RQ where $\phi$ is an $\Ur$-formula, is rewritten as
indicated.
\begin{gather*}
\Forall(x \in \{a \plus A\}, \Exists(y \in \{b \plus A\},\phi(x,y)))
  \why{rule \eqref{forall:iter}} \\
\lfun
\Exists(y \in \{b \plus A\},\phi(a,y)) \\
\qquad{}\land \Forall(x \in A,\Exists(y \in \{b \plus A\} : \phi(x,y)))
  \why{Def. \ref{d:rq}, \eqref{eq:exists}} \\
\lfun
n \in \{b \plus A\} \land \phi(a,n)
\land \Forall(x \in A,\Exists(y \in \{b \plus A\} : \phi(x,y)))
\end{gather*}
Now there are two cases from $n \in \{b \plus A\}$:  $n = b$ and $n \in A$
(rule \eqref{eq:inext}). Let us see the second one:
\begin{gather*}
n \in A \land \phi(a,n) \land \Forall(x \in A,\Exists(y \in \{b \plus A\} : \phi(x,y)))
  \why{rule \eqref{eq:invar}} \\
\lfun
A = \{n \plus N\} \land \phi(a,n) \\
\qquad{} \land \Forall(x \in A,\Exists(y \in \{b \plus A\} : \phi(x,y)))
  \why{rule \eqref{eq:subs}} \\
\lfun
A = \{n \plus N\} \land \phi(a,n)
\land \Forall(x \in \{n \plus N\},\Exists(y \in \{b,n \plus N\} : \phi(x,y)))
\end{gather*}
It is clear that the last $\Forall$ is structurally equal to the initial
formula. Without more information about $\phi$ this could potentially cause an
infinite loop making $\SATRQ$ not to terminate.
\qed
\end{example}

Before presenting the theorems stating termination on different subclasses of
$\Phi_{nrq}$-formulas, consider the following analysis. Let $\phi$ be a
$\Phi_{nrq}$-formula. If $\phi = \phi_1 \lor \phi_2$, then we prove $\SATRQ$
terminates on $\phi_1$ and then on $\phi_2$. Hence, as concerns termination, we
can consider $\phi$ to be a conjunction of $\Phi_{nrq}$-formulas. In this case
$\phi$ can be written as $\phi_\Set \land \phi_\Ur$ where $\phi_\Set$ is a pure
$\Phi_{nrq}$-formula and $\phi_\Ur$ is a $\Ur$-formula. We need to prove
termination of $\SATRQ$ on $\phi_\Set$, as termination of $\SATX$ on $\phi_\Ur$
is guaranteed by the assumption that $\SATX$ is a decision procedure for $\LX$.
Now, if $\phi_\Set$ is a disjunction, we prove termination for each disjunct.
Hence, in the following theorems we prove termination of $\SATRQ$ on
conjunctions of pure $\RQ$-constraints belonging to different subclasses of
formulas.

\begin{teorema}[Termination on $\FU$ formulas]\label{t:termfu} \label{terminationFU}
The $\SATRQ$ procedure can be implemented as to ensure termination
for every conjunction of pure $\FU$ constraints.
\end{teorema}

\begin{proof}

Recall that the only constraint in pure $\FU$ formulas is of the  form
$\Forall(\vec{x}_n \in \vec{A}_n,\phi)$, $0 < n$ and $\phi \in \FX$. First we
will prove that $\SATRQ$ terminates on these constraints. The proof is by
induction on $n$.
\begin{itemize}
\item Base case. Let $\phi$ be a $\LX$ formula. We will show that $\SATRQ$
terminates on the following RUQ:
\begin{equation}\label{eq:basecase}
\Forall(x \in A, \phi(x))
\end{equation}
\begin{itemize}
\item $A = \e$, this case is trivial as rule \eqref{forall:empty} terminates immediately.
\item $A \in \Var_\Set$, this case is trivial as rule \eqref{forall:var} terminates immediately.
\item $A = \{b \plus B\}$, rule \eqref{forall:iter} is applied to \eqref{eq:basecase} yielding:
\begin{equation*}
\phi(b) \land \Forall(x \in B, \phi(x))
\end{equation*}
The recursive call to $\Forall$ is made with a domain strictly smaller than
$A$.  This is so because the call is made with $B$ and because $\phi(b)$ cannot
bind a value to $B$ since $\phi$ is $\Ur$-formula and $B$ is of sort $\sSet$
(the only way of binding a value to $B$ is by means of $t \in B$ or $B = t$,
for some term $t$, which are not generated during $\FU$ processing, Remark
\ref{r:gen}). Then, $\SATRQ$ will terminate when the `end' of $B$ is reached
(i.e., when a variable or the empty set is found).
\end{itemize}
\item Induction hypothesis. $\SATRQ$ terminates on every constraint of the
form $\Forall(\vec{x}_k \in \vec{A}_k, \phi)$ with $k \leq n$, for any
$\Ur$-formula $\phi$.
\item Induction step. Let $\phi$ be any $\Ur$-formula. We will prove that
$\SATRQ$ terminates on the following constraint:
\begin{equation}\label{eq:inductionstep}
\Forall(x \in A, \Forall(\vec{x}_n \in \vec{A}_n,\phi(x,\vec{x})))
\end{equation}
\begin{itemize}
\item $A = \e$, this case is trivial as rule \eqref{forall:empty} terminates immediately.
\item $A \in \Var_\Set$, this case is trivial as rule \eqref{forall:var} terminates immediately.
\item $A = \{b \plus B\}$, rule \eqref{forall:iter} is applied to \eqref{eq:inductionstep} yielding:
\begin{equation*}
\Forall(\vec{x}_n \in \vec{A}_n,\phi(b,\vec{x})) \land \Forall(x \in B, \Forall(\vec{x}_n \in \vec{A}_n,\phi(x,\vec{x})))
\end{equation*}
$\SATRQ$ terminates on the first conjunct by the induction hypothesis.
Besides, the recursive call in the second conjunct is made with a domain
strictly smaller than $A$. This is so because the call is made with $B$ and
because the first conjunct cannot bind a value to $B$ since it is a RUQ or an
$\Ur$-formula and $B$ is of sort $\sSet$. Then, $\SATRQ$ will terminate when
the `end' of $B$ is reached (i.e., when a variable or the empty set are found).
\end{itemize}
\end{itemize}

Observe that termination depends solely on the size of the domain of the RUQ.
If $\phi_\Set$ is a conjunction of RUQ, then termination of $\SATRQ$ on each
RUQ implies termination of $\SATRQ$ for the whole formula. Indeed, when a given
RUQ is processed it can only generate a shorter RUQ or a $\LX$ formula. In
either case, nothing is generated that can bind a value to a domain. Then, the
domain of a RUQ in $\phi_\Set$ is not affected by the processing of the other
RUQ in the formula. \qed
\end{proof}

Before proving termination of $\SATRQ$ on $\FE$ formulas we need the following
lemma.

\begin{lemma}\label{l:termnorq}
$\SATRQ$ terminates on any $\RQ$-formula without RQ.
\end{lemma}

\begin{proof}
Let $\phi$ be a $\RQ$-formula without RQ. Write $\phi$ as $\phi_\Set \land
\phi_\Ur$. Hence, $\phi_\Set$ is comprised solely of membership and equality
constraints. Then, only the \textsc{Membership} and \textsc{Equality} rules of
Figure \ref{f:rules} will be used by $\SATRQ$. These rules have been proved to
constitute a terminating rewriting system elsewhere \cite[Theorem
10.10]{Dovier00}. \qed
\end{proof}

\begin{teorema}[Termination on $\FE$ formulas]\label{t:termfe} \label{terminationFE}
The $\SATRQ$ procedure can be implemented as to ensure termination
for every conjunction of pure $\FE$ constraints.
\end{teorema}

\begin{proof}
Let $\phi$ be an $\Ur$-formula. Consider the following rewriting:
\begin{gather*}
\Exists([x \in A, y \in B], \phi(x,y))
  \why{Def. \ref{d:rq}, \eqref{eq:exists}} \\
\lfun n_1 \in A \land \Exists(y \in B, \phi(n_1,y))
  \why{Def. \ref{d:rq}, \eqref{eq:exists}} \\
\lfun n_1 \in A \land n_2 \in B \land \phi(n_1,n_2)
\end{gather*}
Then, all REQ are quickly eliminated from the formula. This can be  easily
generalized to $\Exists(\mathbf{x}_n \in \mathbf{A}_n,\phi)$ for any
$\Ur$-formula $\phi$ and any $0 < n$. The resulting non-$\Ur$ subformula is a
$\RQ$-formula without RQ. Hence, by Lemma \ref{l:termnorq}, $\SATRQ$ terminates
on that formula. Given that conjunctions of REQ are rewritten into conjunctions
of formulas such as the last one above, $\SATRQ$ terminates on every
conjunction of pure $\FE$ formulas. \qed
\end{proof}

\begin{teorema}[Termination on $\FEU$ formulas] \label{terminationFEU}
The $\SATRQ$ procedure can be implemented as to ensure termination
for every conjunction of pure $\FEU$ constraints.
\end{teorema}

\begin{proof}
Recall that the only constraints in pure $\FEU$ formulas are of the form:
\begin{equation}\label{eq:feuconstraint}
\Exists(\vec{x}_k \in \vec{A}_k,\Forall(\vec{y}_m \in \vec{B}_m, \phi))
\end{equation}
for some $\Ur$-formula $\phi$, $0 < k$ and $0 \leq m$ (if $m = 0$, then $\phi$
is the innermost filter).

First we prove termination on such a constraint. By using the same reasoning of
Theorem \ref{t:termfe}, \eqref{eq:feuconstraint} is rewritten into:
\begin{equation}\label{eq:t51}
n_1 \in A_1 \land \dots \land n_k \in A_k \land \Forall(\vec{y}_m \in \vec{B}_m, \phi)
\end{equation}

Then, $\SATRQ$ process $n_1 \in A_1 \land \dots \land n_k \in A_k$. By Lemma
\ref{l:termnorq}, $\SATRQ$ terminates on that conjunction. The processing of
this conjunction either terminates in $\false$, and so $\SATRQ$ stops, or it
yields a conjunction of the form:
\begin{equation}\label{eq:t52}
(\bigwedge_{i = 1}^{v} X_i = t_i) \land (\bigwedge_{i = 1}^{w} Y_i = u_i)
\end{equation}
where $X_i \in \Var_\Set$, $Y_i \in \Var_\Ur$, $t_i$ are set terms and $u_i$
are $\Ur$ terms.

If some $B_j$ in $\vec{B}$ is either $X_i$ or $\{\dots\plus X_i\}$, then $X_i$
is substituted by $t_i$. This rewrites $\Forall(\vec{y}_m \in \vec{B}_m, \phi)$
into $\Forall(\vec{y}_m \in \vec{B}'_m, \phi)$. Then we have the following
formula:
\begin{equation}\label{eq:t53}
(\bigwedge_{i = 1}^{v} X_i = t_i) \land (\bigwedge_{i = 1}^{w} Y_i = u_i) \land \Forall(\vec{y}_m \in \vec{B}'_m, \phi)
\end{equation}

By Theorem \ref{t:termfu} $\SATRQ$ terminates on $\Forall(\vec{y}_m \in
\vec{B}'_m, \phi)$.  While processing the RUQ, $\SATRQ$ can only generate
either RUQ or $\Ur$-formulas (Remark \ref{r:gen}). Then, the main loop of
$\SATRQ$ terminates. It only remains to call $\SATX$ on the $\Ur$-subformula,
which will terminate under the assumption that it is a decision procedure.
Hence, $\SATRQ$ terminates on \eqref{eq:feuconstraint}.

Now we prove that $\SATRQ$ terminates on a conjunction of constraints such as
\eqref{eq:feuconstraint}. We can think that $\SATRQ$ will process each such
constraint by going through formulas \eqref{eq:t51}-\eqref{eq:t53}. Processing
each final RUQ can only generate either RUQ or $\Ur$-formulas. Then, the main
loop of $\SATRQ$ terminates. \qed
\end{proof}

Now we consider a more general subclass of $\Phi_{nrq}$-formulas which,
however, must obey a restriction concerning the domains of REQ that go
\emph{after} RUQ in mixed RQ (some times called \emph{alternating quantifiers}
\cite{DBLP:journals/lmcs/FeldmanPISS19}).

We say that an RQ has a \emph{variable domain} if its domain is either a
variable or an extensional set whose set part is a variable. Moreover, we say
that the variable of the domain is the domain itself (if it is a variable) or
its set part. For example, in:
\begin{equation}\label{eq:rqex}
\Forall(x \in \{h \plus \dot{A}\}, \Exists(y \in \dot{B}, \phi(x,y)))
\end{equation}
$\{h \plus \dot{A}\}$ is a variable domain whose variable is $A$, and $\dot{B}$
is variable domain whose variable is $B$.

The class of formulas we are about to define will avoid formulas such as the
one in Example \ref{ex:undec}. The problem with that formula is that there is
an $\Exists$ constraint \emph{after} a $\Forall$ constraint with the same
domain variable ($A$). In this situation when the $\Forall$ constraint picks an
element ($a$) of its domain the $\Exists$ constraint hypothesizes the existence
of a new element ($n$) in $A$ as to satisfy $\phi$. As now $n \in A$, the
$\Forall$ constraint must pick $n$ making the $\Exists$ constraint to
hypothesize the existence of another new element in $A$. This behavior may
produce an infinite rewriting loop. In a sense, the $\Exists$ constraint feeds
back the $\Forall$ constraint with new elements \emph{if they have the same
domain variable}. This problem can be generalized to conjunctions of RQ.

\begin{example}\label{ex:undec2}
The following formula:
\begin{gather*}
\Forall(x \in \{h \plus \dot{A}\}, \Exists(y \in \dot{B}, \phi(x,y))) \\
{}\land
\Forall(z \in \dot{B}, \Exists(w \in \{b \plus \dot{A}\},\psi(z,w)))
\end{gather*}
may produce an infinite loop even though the $\Forall$ and $\Exists$
constraints sharing the same domain variable ($A$) are in different RUQ. Still,
in a sense, the $\Exists$ constraint with domain variable $A$ is \emph{after}
the $\Forall$ constraint with the same domain variable: from the $A$ in the
$\Forall$ we go to the $B$ in the $\Exists$, from this we go to the $B$ in the
$\Forall$ which leads us to the $A$ in the $\Exists$. \qed
\end{example}

Therefore, the mathematics we are going to define are meant to characterize
formulas such as those in Examples \ref{ex:undec} and \ref{ex:undec2}.

Let $\phi_1 \land \dots \land \phi_n$ be a conjunction of nested RQ. Each RQ in
a nested RQ is indexed by its position in the chain. For instance, in
\eqref{eq:rqex} the $\Forall$ constraint has index 1 while the $\Exists$
constraint has index 2. For each $\phi_i$ build the function, called
\emph{domain function of $\phi_i$}, whose ordered pairs are of the form
$((i,j),(\mathsf{Q}_j^i,D_j^i))$ where:
\begin{itemize}
\item A pair with first component $(i,j)$ is in the domain function
of $\phi_i$ iff the RQ with index $j$ in $\phi_i$ has a variable domain.
\item $D_j^i$ is the domain variable of the RQ with index $j$ in $\phi_i$.
\item $\mathsf{Q}_j^i$ is $\forall$ if the $j$ RQ is a $\Forall$ constraint
and is $\exists$ if the RQ is an $\Exists$ constraint, in $\phi_i$.
\end{itemize}
Hence, the domain function of the formula of Example \ref{ex:undec2} is:
\begin{equation*}
\{((1,1),(\forall,A)), ((1,2),(\exists,B)),((2,1),(\forall,B)), ((2,2),(\exists,A))\}
\end{equation*}

From the domain functions build a directed graph, called \emph{domain graph},
whose nodes are the ordered pairs of the domain functions. The edges are built
as follows:
\begin{enumerate}
\item If $((i,j),(\forall,D_{j}^i))$ and $((i,k),(\exists,D_{k}^i))$, with $j < k$,
are in a domain function, then $((i,j),(\forall,D_{j}^i)) \fun
((i,k),(\exists,D_{k}^i))$ is an edge of the domain graph.
%
\item If $((i,j),(\exists,D))$ and $((b,a),(\forall,D))$ are in domain functions
with $i \neq b$, then $((i,j),(\exists,D)) \fun ((b,a),(\forall,D))$ is an edge
of the domain graph.
\end{enumerate}
Hence, the domain graph of the formula of Example \ref{ex:undec2} is:
\begin{gather*}
((1,1),(\forall,A)) \fun ((1,2),(\exists,B)), ((2,1),(\forall,B)) \fun ((2,2),(\exists,A)), \\
((1,2),(\exists,B)) \fun ((2,1),(\forall,B)), ((2,2),(\exists,A)) \fun ((1,1),(\forall,A))
\end{gather*}

\newcommand{\fe}{\forall\hspace{-2pt}\exists}

Consider the domain graph of a conjunction of pure $\RQ$-constraints. A path in the graph such as:
\begin{gather*}
((i_1,j_{11}),(\forall,D)) \fun ((i_1,j_{12}),(\exists,D_1)) \;\fun\\
((i_2,j_{21}),(\forall,D_1)) \fun ((i_2,j_{22}),(\exists,D_2)) \fun\\
\cdots\cdots\cdots\cdots\cdots\cdots
\cdots\cdots\cdots\cdots\cdots\cdots\cdot\cdot\fun \\
((i_n,j_{n1}),(\forall,D_{n-1})) \fun ((i_n,j_{n2}),(\exists,D))
\end{gather*}
where $i_a \neq i_b$ if $a \neq b$ for all $a,b \in [1,n]$, is called a $\fe$
\emph{loop}.  Note that the first and last domain variables in a $\fe$ loop are
the same ($D$).

We say that a conjunction of RQ is free of $\fe$ loops if there are no $\fe$
loop in its domain graph. $\Phi_{\fe}$ is the set of conjunctions of RQ free of
$\fe$ loops. Note that $\Phi_{\fe}$ includes all the $\Phi_{nrq}$-formulas
where there are no variable domains.

\begin{teorema}[Termination on $\Phi_{\fe}$ formulas] \label{t:termFUE}
The $\SATRQ$ procedure can be implemented as to ensure termination
for every formula in $\Phi_{\fe}$.
\end{teorema}

\begin{proof}
First we prove that $\SATRQ$ terminates on an atomic formula, $\psi$.  It
starts by removing from $\psi$ all the leading $\Exists$ constraints (if any)
and then proceeding as in Theorem \ref{terminationFEU}. Hence, we get a formula
such as \eqref{eq:t53}:
\begin{equation*}
(\bigwedge_{i = 1}^{v} X_i = t_i) \land (\bigwedge_{i = 1}^{w} Y_i = u_i) \land \Forall(\vec{y}_m \in \vec{B}'_m, \Phi)
\end{equation*}
but where $\Phi$ is a mixed RQ (eventually ending in a $\Ur$-formula) whose
domains might have been changed during the substitution step (see proof of
Theorem \ref{terminationFEU}).

Now, $\SATRQ$ processes the $\Forall$ constraint as in Theorem \ref{t:termfu}.
Here, though, we cannot easily conclude that RUQ processing cannot bind a value
to $\vec{B}'$ because after the leading RUQ there might be some REQ. The
problem with REQ is that they generate constraints of the form $n \in A$, for
some domain $A$. Then, if the variable of $A$ happens to
be a variable in $\vec{B}'$ (or in some RUQ in $\Phi$) we will have an infinite
loop as in Examples \ref{ex:undec} and \ref{ex:undec2}. However, since $\psi$
belongs to $\Phi_{\fe}$, we know that there is no domain variable shared
between a RUQ and a REQ ahead of it because $\psi$ is free from $\fe$ loops.
Hence, we can arrive at the same conclusion of Theorem \ref{t:termfu} meaning
that $\SATRQ$ terminates on $\psi$.

Now we prove that $\SATRQ$ terminates on a conjunction of constraints such as
$\psi$. Again, all the leading $\Exists$ constraints (if any) are removed from
each $\psi_i$ thus generating a formula such as \eqref{eq:t53} but with a
conjunction of $\Forall$ constraints\footnote{$\Forall_i$ means that all its
elements are renamed accordingly: $\vec{y}_{i,m_i}$, $\vec{B}'_{i,m_i}$,
$\Phi_i$.}:
\begin{equation*}\label{eq:tmix1}
(\bigwedge_{i = 1}^{v} X_i = t_i) \land (\bigwedge_{i = 1}^{w} Y_i = u_i) \land (\bigwedge_{i = 1}^{q} \Forall_i(\vec{y}_m \in \vec{B}'_m, \Phi))
\end{equation*}
with $0 \leq v,w,q$.

As above, $\SATRQ$ processes all the $\Forall_i$ (as in Theorem \ref{t:termfu})
and, again, the problem are the $n \in A$ constraints ($A$ a variable domain)
that might be generated by the possible REQ present in each $\Phi_i$.
Differently from the base case, here a $n \in A$ constraint generated when
$\Forall_a$ is processed might affect a domain variable of a $\Forall$
constraint present in $\Forall_b$ with $a \neq b$, as shown in Example
\ref{ex:undec2}. However, since the formula belongs to $\Phi_{\fe}$ we know is
free from $\fe$ loops. This includes loops starting with a $\Forall$
domain variable in $\Forall_a$ and ending with the same domain variable in an
$\Exists$ constraint present in $\Forall_b$. Then, no $n \in A$ constraint can
affect a $\Forall$ domain variable. Therefore, $\SATRQ$ terminates. \qed
\end{proof}

We close this section with the following two observations.

\begin{remark}
$\SATRQ$ terminates for some formulas in $\Phi_{nrq} \setminus \Phi_{\fe}$. For
example,  the following is a slight variation of the formula of Example
\ref{ex:undec}:
\begin{equation*}
\Forall(x \in \dot{A}, \Exists(y \in \{b \plus \dot{A}\},\phi(x,y)))
\end{equation*}
for which $\SATRQ$ trivially terminates because the formula is irreducible
given that the domain is a variable (rule \eqref{forall:var} applies). We could
have tighten the definition of $\Phi_{\fe}$ as to include this kind of
formulas. However, we consider that these formulas do not constitute a proper
subclass as termination depends on whether or not some RUQ remain irreducible
throughout the constraint solving procedure. \qed
\end{remark}

\begin{remark}\label{r:restrictions}
$\LRQ$ has been designed by imposing some restrictions on its fundamental
elements,  namely: $\Var_\Set \cap \Var_\Ur = \emptyset$  and $\Pi_\Set \cap
\Pi_\Ur = \emptyset$. These restrictions are used to prove Theorems
\ref{satisf}-\ref{t:termFUE}. However, they can be relaxed to some extent as to
accept a wider class of theories as the parameter for $\LRQ$.

It is possible to accept an $\LX$ such that $\Var_\Set \cap \Var_\Ur \neq
\emptyset$  and $\Pi_\Set \cap \Pi_\Ur \neq \emptyset$ provided the solutions
returned by $\SATX$ are compatible with the irreducible form of Definition
\ref{irreducible}. That is, if $\SATX$ returns, as part of its solutions, a
conjunction of constraints including set variables, this conjunction must be
satisfiable by substituting all set variables by the empty set. This would be
enough as to prove Theorem \ref{satisf}.

Along the same lines, if $\Pi_\Ur$ contains $\in$ then termination of $\SATRQ$
might be compromised as $\SATX$ might generate $\in$-constraints where the
right term is the domain of a $\Forall$ (much as when an $\Exists$ is after a
$\Forall$). This can be generalized to any predicate symbol in $\Pi_\Ur$ that
can bind values to set terms. In this case Theorems
\ref{t:termfu}-\ref{t:termFUE} can be proved if $\Forall$ domains are not
affected by the constraints generated by $\SATX$.

As we will shown in Section \ref{practice} there are expressive $\LX$ such
that $\Var_\Set \cap \Var_\Ur \neq \emptyset$ and $\Pi_\Set \cap \Pi_\Ur \neq
\emptyset$ for which $\SATRQ$ is still an effective solver. \qed
\end{remark}

\section{\label{extensions}Avoiding Existential Variables Inside RQ}

All the considerations made in this section concerning RUQ apply equally to
REQ. The concepts introduced here are adapted from those developed by us for
RIS \cite[Section 6]{DBLP:journals/jar/CristiaR21a}.

Assume $R$ is a set of ordered pairs. We can try to write a formula stating
that $R$ is the identity relation:
\begin{equation}\label{eq:ex}
\Forall(x \in R, x = (e,e))
\end{equation}
where $e$ is intended to be a variable existentially quantified inside the RUQ.
As defined in Section \ref{lrq}, $\LRQ$ does not allow to introduce these
variables and so \eqref{eq:ex} would not be one of its formulas. Besides, it is
not clear what a quantification domain could be for $e$. However, the following
is an $\RQ$-formula stating the same property:
\begin{equation*}
\Forall((x,y) \in R, x = y)
\end{equation*}
Note that we use an ordered pair as the control term (see Definition
\ref{RQ-terms}). Precisely, allowing (nested) ordered pairs as control terms
makes it possible to avoid many existential variables inside RUQ. As binary
relations are a fundamental concept in Computer Science
\cite{DBLP:journals/jar/Givant06,DBLP:conf/RelMiCS/BerghammerHS14}, the
introduction of ordered pairs as control terms is sensible as it enables to
quantify over binary relations, without introducing existential variables.

Even if we allow (unrestricted) existential variables inside RQ, it is
important to avoid them because the negation of such a RQ would not be a $\LRQ$
formula. Indeed, if there are existential variables inside the RQ the negation
will introduce a universally quantified formula, which is not a $\LRQ$ formula.
For instance,
assuming $R$ is a set of numeric ordered pairs, the following is a predicate
stating that the sum of any of its elements is greater than $z$:
\begin{equation}\label{eq:sum}
\Forall((x,y) \in R, sum(x,y,n) \land z < n)
\end{equation}
where $sum(x,y,n)$ is interpreted as $n = x + y$, and $n$ is intended to
be a variable existentially quantified inside the RUQ. Here,
control terms do not help and, again, we do not have a quantification domain
for $n$. Then, the negation of this formula would inevitably introduce a
universal quantification. Furthermore,  it would negate $sum(x,y,n)$ for all
$n$ which would mean that there is no result for $x+y$.

Hence, the language is extended by introducing a $\Forall$ constraint of arity
4:
\begin{equation}\label{eq:ruq4}
\Forall(x \in A,[e_1,\dots,e_n],\phi(x,e_1,\dots,e_n),\psi(x,e_1,\dots,e_n))
\end{equation}
where $e_1,\dots,e_n$ are variables implicitly existentially quantified inside
the RUQ and $\psi$ is a conjunction of so-called \emph{functional predicates}.
A predicate $p$ of artity $n+1$ ($0 < n$) is a functional predicate iff for
each $x_1,\dots,x_n$ there exists exactly one $y$ such that
$p(x_1,\dots,x_n,y)$ holds; $y$ is called the \emph{result} of $p$. For
instance, $sum$ is a functional predicate. In an extended RUQ, $e_1,\dots,e_n$
must be the results of the functional predicates in $\psi$.

The semantics of \eqref{eq:ruq4} is:
\begin{equation*}
\forall x(x \in A \implies (\exists e_1,\dots,e_n (\psi(x,e_1,\dots,e_n) \land \phi(x,e_1,\dots,e_n))))
\end{equation*}
whereas its negation is:
\begin{equation*}
\exists x(x \in A \land (\exists e_1,\dots,e_n (\psi(x,e_1,\dots,e_n) \land \lnot\phi(x,e_1,\dots,e_n))))
\end{equation*}
given the functional character of $\psi$ \cite[Section
6.2]{DBLP:journals/jar/CristiaR21a}. By means of functional predicates the
introduction of existential variables inside RUQ is harmless while the
expressiveness of the language is widened.

\begin{example}
Formula \eqref{eq:sum} should be written by means of an extended RUQ:
\begin{equation*}
\Forall((x,y) \in R,[n],z < n,sum(x,y,n))
\end{equation*}
Note that $n$ is the result of $sum$. The negation of the above formula is:
\begin{equation*}
\Exists((x,y) \in R,[n],z \geq n,sum(x,y,n))
\end{equation*}
which is consistent with the intended meaning of the original formula.
\qed
\end{example}

\section{\label{practice}$\LRQ$ in Practice}

$\LRQ$ and $\SATRQ$ have been implemented as part of the \setlog (`setlog')
tool \cite{setlog}. \setlog is a constraint logic programming (CLP) language
implemented in Prolog. It also works as a satisfiability solver (and thus as an
automated theorem prover) for a few theories rooted in the theory of finite
sets. \setlog and the theories underlying it have been thoroughly described
elsewhere
\cite{Dovier00,DBLP:journals/tplp/CristiaRF15,DBLP:journals/jar/CristiaR20,DBLP:journals/jar/CristiaR21a,cristia_rossi_2021,zbMATH07552282,DBLP:journals/corr/abs-2105-03005}.
Empirical evidence of the practical capabilities of \setlog has been provided
as well
\cite{CristiaRossiSEFM13,DBLP:journals/jar/CristiaR21,DBLP:journals/jar/CristiaR21b,10.1093/comjnl/bxab030,DBLP:journals/corr/abs-2112-15147}.

Theory $\Ur$ in \setlog is the combination between the theories known as
$\mathcal{LIA}$ and $\mathcal{BR}$. $\mathcal{LIA}$ stands for \emph{linear
integer arithmetic} and implements a decision procedure for systems of linear
equations and disequations over the integer numbers. $\mathcal{BR}$ stands for
\emph{binary relations} and implements a decision procedure for an expressive
fragment of finite set relation algebra (RA)
\cite{DBLP:journals/jar/CristiaR20}. In $\mathcal{BR}$, binary relations are
sets of ordered pairs and all the RA operators are available as constraints,
namely: union ($C = A \cup B \fun$ \verb+un(A,B,C)+), intersection ($C = A \cap
B \fun$ \verb+inters(A,B,C)+), identity relation over a set ($\id(A) = R
\fun$\verb+id(A,R)+), converse of a binary relation ($R^\smile = S \fun$
\verb+inv(R,S)+) and composition ($T = R \circ S \fun$ \verb+comp(R,S,T)+).
These operators can be combined in $\mathcal{BR}$-formulas to define many other
operators such as: domain ($\dom R = A \fun$ \verb+dom(R,A)+) and range ($\ran
R = A \fun$ \verb+ran(R,A)+) of a binary relation, a predicate constraining a
binary relation to be a function (\verb+pfun(R)+), function application ($F(X)
= Y \fun$ \verb+applyTo(F,X,Y)+), etc.

\newcommand{\LRQBR}{\LRQ(\mathcal{LIA}+\mathcal{BR})}

$\mathcal{LIA}$ satisfies all the restrictions discussed in Remark
\ref{r:restrictions}, but $\mathcal{BR}$ does not. However, the solutions
returned by \setlog when solving $\mathcal{BR}$-formulas are compatible with
the irreducible form of Definition \ref{def:solved} \cite[Definition 15 and
Theorem 3]{DBLP:journals/jar/CristiaR20}. That is, irreducible formulas in
$\mathcal{BR}$ are satisfied by substituting all the set and relational
variables by the empty set. However, when a $\mathcal{BR}$-formula is
processed, RQ domains may be affected. As discussed in Remark
\ref{r:restrictions} this may compromise termination; $\LRQBR$ decidability is
discussed more deeply in Section \ref{relwork}. Nevertheless, as the following
case studies show, \setlog is still an effective and efficient tool to
automatically reason about $\LRQBR$ formulas. Termination do not seem to be an
issue for many classes of practical problems expressible in $\LRQBR$.

\setlog's concrete syntax is a slight variation of the syntax used in this
paper: $\{\_\plus\_\}$ is \verb+{_/_}+; $(\_,\_)$ is \verb+[_,_]+; $\in$ is
\verb+in+; $\land$ is \verb+&+; $A \subseteq B$ is \verb+subset(A,B)+;
variables begin with a capital letter.

The following simple example shows \setlog syntax and how to use it to prove
invariance lemmas.

\begin{example}
Let \verb+Usr+ and \verb+Admin+ be the sets of users and administrators of some
system. Let us say that the security policy requires these sets to be disjoint.
We can express that in \setlog as follows\footnote{Other encodings are
possible; we deliberately choose to use a RUQ.}:
\begin{verbatim}
inv(Usr,Adm) :- foreach([U in Usr,A in Adm], U neq A).
\end{verbatim}
We can model the operation adding user \verb+X+ to \verb+Usr+  yielding \verb+Usr_+ as the new set:
\begin{verbatim}
addUsr(Usr,Adm,X,Usr_,Adm_) :- Usr_ = {X / Usr} & Adm_ = Adm.
\end{verbatim}
We would like to know if \verb+addUsr+ preserves \verb+inv+, so we run the
following query\footnote{Given that \setlog is a satisfiability solver we call
it on the negation of the lemma waiting for a \texttt{no} (i.e., $\false$)
answer.}:
\begin{verbatim}
neg(inv(Usr,Adm) & addUsr(Usr,Adm,X,Usr_,Adm_) implies inv(Usr_,Adm_)).
\end{verbatim}
As \verb+addUsr+ fails to preserve \verb+inv+, \setlog provides a
counterexample (\verb+N+ new variable):
\begin{verbatim}
Admin = {X / N}, Usr_ = {X / Usr}, Admin_ = {X / N}
\end{verbatim}
So we can fix \verb+addUsr+ by adding a pre-condition:
\begin{verbatim}
addUsr(Usr,Adm,X,Usr_,Adm_) :- X nin Usr & Usr_ = {X / Usr} & Adm_ = Adm.
\end{verbatim}
Now the answer to the query is \verb+no+ meaning that the formula is unsatisfiable.
\qed
\end{example}

The next three subsections present real-world case studies where \setlog
is used as a CLP language and as an automated verifier. The focus in on how RQ
are used.

\subsection{The Landing Gear System}

In the fourth edition of the ABZ Conference held in Toulouse (France) in 2014,
Boniol and Wiels proposed a real-life, industrial-strength case study, known as
the Landing Gear System (LGS) \cite{DBLP:conf/asm/BoniolW14}. Mammar and Laleau
\cite{DBLP:conf/asm/MammarL14} developed an Event-B
\cite{Abrial:2010:MES:1855020} specification of the LGS, formally verified
using Rodin \cite{DBLP:journals/sttt/AbrialBHHMV10}, ProB \cite{Leuschel00} and
AnimB\footnote{\url{http://www.animb.org}}. Basically, we encoded in \setlog
the Event-B specification and used \setlog to automatically discharge all the
proof obligations generated by  Rodin. This work is thoroughly described
elsewhere \cite{DBLP:journals/corr/abs-2112-15147}\footnote{\setlog code of the
LGS: \url{http://www.clpset.unipr.it/SETLOG/APPLICATIONS/lgs.zip}.}.

This is the simplest model in terms of RQ as it does not require nested RQ. A
typical use of RQ in the LGS is the following\footnote{Some variable names are
changed to save some space.}:
\begin{verbatim}
ta_inv5(Positions,DULDC) :-
  pfun(DULDC) &
  dompf(DULDC,Positions) &
  foreach([X,Y] in DULDC, 0 =< Y).
\end{verbatim}
That is, \verb+ta_inv5+ defines a state invariant corresponding to the Event-B
machine named TimedAspects. In mathematical notation the invariant states
$DULDC \in Positions \pfun \nat$. As in \setlog we cannot express $\nat$ we use
a RUQ to ascertain that the second component of each element in $DULDC$ is
non-negative. Then, invariance lemmas such as:
\begin{verbatim}
neg(di_inv1(Positions,Dcp) &
    ta_inv1(CT) &
    ta_inv5(Positions,DULDC) &
    ta_make_DoorClosed(...,Dcp,...,CT,...,DULDC,...,Dcp_,...,DULDC_)
    implies   ta_inv5(Positions,DULDC_)).
\end{verbatim}
are automatically discharged by \setlog. \verb+di_inv1(Positions,Dcp)+ and
\verb+ta_inv1(CT)+ are other invariants that are needed as hypothesis and
\verb+ta_make_DoorClosed+ is one of the state transitions of the LGS
model (ellipses stand for variables). Hence, \verb+ta_make_DoorClosed+ changes
the value of \verb+DULDC+ during the state transition and so we need to check
that \verb+ta_inv5+ is still valid in the new state.

\setlog discharges all the 465 proof obligations in less than 5 minutes.

\subsection{The Bell-LaPadula Security Model}

Around 1973 D.E. Bell and L. LaPadula published the first formal model of a
secure operating system \cite{BLP1,BLP2}. Today this model is known as the
Bell-LaPadula model, abbreviated as BLP. BLP is described as a state machine by
means of first-order logic and set theory. The model also formalizes two state
invariants known as \emph{security condition} and \emph{*-property}. We encoded
BLP and its properties in \setlog and used it to automatically discharge all
the invariance lemmas. This work is presented with detail elsewhere
\cite{DBLP:journals/jar/CristiaR21}\footnote{\setlog code of BLP:
\url{http://www.clpset.unipr.it/SETLOG/APPLICATIONS/blp2.zip}.}.

The following is the \setlog encoding of the *-property:
\begin{verbatim}
starProp(State) :-
  State = [[br,Br],[bw,Bw],[fo,Fo],[fs,Fs],[m,M]] &
  foreach([[S1,O1] in Br, [S2,O2] in Bw], [Sco1,Sco2],
    S1 = S2 implies dominates(Sco1,Sco2),
    applyTo(Fo,O1,Sco1) & applyTo(Fo,O2,Sco2)).
\end{verbatim}
As can be seen it requires the use of the extended version of nested RUQ
(Section \ref{extensions}). That is, it declares two existential variables
inside the RUQ (\verb+Sco1+ and \verb+Sco2+) and uses the functional predicate
section (\verb+applyTo(Fo,O1,Sco1) & ...+). Extended RUQ are also used in the
state transitions, for instance:
\begin{verbatim}
getRead(State,S,O,State_) :-
  State = [[br,Br],[bw,Bw],[fo,Fo],[fs,Fs],[m,M]] &
  [O,[S,read]] in M & [S,O] nin Br &
  applyTo(Fo,O,Sco) & applyTo(Fs,S,Scs) &
  dominates(Sco,Scs) &
  foreach([Si,Oi] in Bw,[Scoi],
    Si = S implies dominates(Sco,Scoi),applyTo(Fo,Oi,Scoi)) &
  Br_ = {[S,O]/Br} &
  State_ = [[br,Br_],[bw,Bw],[fo,Fo],[fs,Fs],[m,M]].
\end{verbatim}
\verb+getRead+ grants \verb+read+ permission to subject \verb+S+ on object
\verb+O+ in which case changes the value of variable \verb+Br+. Then, the
following invariance lemma must be proved:
\begin{verbatim}
neg(starProp(State) & getRead(State,S,O,State_) implies starProp(State_)).
\end{verbatim}

Due to an optimization introduced since our first experiments with BLP, now
\setlog proves all the 60 invariance lemmas in less than 2 seconds instead of
the 11.5 seconds reported previously \cite{DBLP:journals/jar/CristiaR21}.

\subsection{\label{android}Android's Permission System}

In a series of articles a group of Uruguayan and Argentinian researchers and
students developed a certified Coq model and implementation of Android's
permission system
\cite{DBLP:conf/ictac/BetarteCLR15,DBLP:journals/cuza/BetarteCLR16,DBLP:journals/cleiej/LunaBCSCG18,DBLP:conf/types/Luca020}.
They model the system as a state machine, then propose a number of properties
and use Coq to verify them against the model. Properties are classified in two
classes: valid state properties and security properties. The first class
ensures the state machine preserves some well-formedness properties of the
state variables, while the second ensures Android behaves as expected in some
security-related scenarios. As with the previous case studies, we translated
the Coq model into \setlog and used it to automatically prove
properties\footnote{\setlog code of Android 10's permission system:
\url{http://www.clpset.unipr.it/SETLOG/APPLICATIONS/android.zip}}.

This is the most challenging case study we have developed so far. It takes
\setlog to its limits concerning reasonable computing times to discharge proof
obligations\footnote{Actually, \setlog is unable to prove only three of the
properties proposed in the Coq model.}. It also uses the most complex nested RQ
we have used so far, as the following one which formalizes one of the valid
state properties\footnote{In the following formulas some simplifications are
introduced to save some space.}.
\begin{verbatim}
notDupPerm(DP) :-
  foreach([[A1,SP1] in DP, [A2,SP2] in DP],
    foreach([P1 in SP1, P2 in SP2],[IP1,IP2],
      IP1 = IP2 implies P1 = P2 & A1 = A2,
      idP(P1,IP1) & idP(P2,IP2))).
\end{verbatim}
Note that there is one nested RUQ whose filter is an extended nested RUQ.
Furthermore, in the innermost RUQ \verb+SP1+ and \verb+SP2+ are domains whereas
they are part of the outermost control term. That is, first \verb+[A1,SP1]+
quantifies over \verb+DP+ and then \verb+P1+ quantifies over \verb+SP1+.

The following is another state consistency property fitting in the $\Phi_{\fe}$ subclass.
\begin{verbatim}
permsDom(PR,Apps,SS) :-
  foreach([A,P] in PR,
    A in Apps or exists(SI in SS, [IA], IA = A, idSI(SI,IA))).
\end{verbatim}
As can be seen, the domain of the RUQ does not appear inside the REQ, thus
making the formula free of $\fe$ loops. The REQ uses a functional predicate.

\setlog automatically discharges 801 proof obligations in around 22 minutes.

\section{\label{relwork}Discussion and Related Work}

\newcommand{\CL}{\forall_{0,2}^\pi}

The problem of deciding the satisfiability of quantified formulas is obviously
undecidable. Hence, we can only hope to find expressive fragments that are
decidable. A possible path for this is to restrict the form of the quantifiers
and another is to allow only certain quantifier-free formulas. A class of
quantified formulas that has been studied for many years is that of restricted
quantifiers. However, the full fragment of restricted quantifiers as well as
some of its sub-fragments are undecidable (e.g.
\cite{https://doi.org/10.1002/cpa.3160340203,https://doi.org/10.1002/cpa.3160460104,DBLP:journals/lmcs/FeldmanPISS19}).
Hence, further restrictions must be imposed. One of such restrictions is to
deal with formulas where RUQ are after REQ. This fragment has been proved to be
decidable in different contexts
\cite{https://doi.org/10.1002/cpa.3160340203,DBLP:journals/jsyml/OmodeoP12,DBLP:journals/lmcs/FeldmanPISS19}.
We started by working with quantifier-free formulas that do not affect
quantification domains (Definition \ref{signature}, Remark
\ref{r:restrictions}). We have not seen other works taking this path. Then, we
relaxed that restriction loosing termination but nonetheless gaining
expressiveness. This combination seems to be useful in practice (Section
\ref{practice}). 

We identify Computable Set Theory (CST)
\cite{10.5555/92143,DBLP:series/mcs/CantoneOP01} as the main and closest source
of works related to the one presented in this paper. However, there are works
outside CST dealing with similar problems, specially in the realm of STM
solvers. We start with the latter.

The STM solving community deals with unrestricted quantifiers.  The usual
practical technique employed in SMT solvers to deal with quantified formulas is
heuristics-based quantifier instantiation
\cite{Flanagan2003,Dutertre01,Moura2007,Ge2009}. In particular, Simplify's
E-matching algorithm \cite{Detlefs2005} is used by some of these tools.
Heuristic instantiation manages to solve problems of software verification.
However, it suffers from some shortcomings as stated by Ge and de Moura
\cite{DBLP:conf/cav/GeM09}. For this reason, Ge and de Moura propose some
decidable fragments of first order logic modulo theories. The proposed decision
procedures can solve complex quantified array properties. The authors show how
to construct models for satisfiable quantified formulas in these fragments.

In a more recent work, Feldman et al. \cite{DBLP:journals/lmcs/FeldmanPISS19}
study the problem of discharging inductive invariants with quantifier
alternation using SMT solvers. They depart from formulas belonging to the
Effectively Propositional logic (EPR), also known as the
Bernays-Schönfinkel-Ramsey class. In this logic, formulas are of the form
$\exists^{*}\hspace{-2pt}\forall^*(\delta)$, where $\delta$ is a
quantifier-free formula over some first-order vocabulary. This logic has been
proved to be decidable and useful in automatically discharging verification
conditions of software involving linked-lists, distributed protocols, etc.
Feldman and his colleagues then go to extend EPR with formulas of the form
$\forall^{*}\hspace{-2pt}\exists^*(\delta)$. The first conclusion they get is
that this fragment is undecidable. However, a second conclusion is that some
techniques can be put to work as to solve many interesting problems in that
fragment. The main technique is instantiations that are bounded in the depth of
terms. However, bounded instantiations guarantee termination a-priori even when
the invariant is not correct. In these cases the algorithm returns an
approximated counterexample. The invariants approached by Feldman at al. are of
the same form of most of the proof obligations present in our case studies.

As can bee seen, the SMT solving community approaches the problem of  finding
decision procedures for quantified fragments of logic languages in a quite
different way as we do. They do not use RQ nor a theory of sets. RQ have an
interesting property: $\forall x \in A: \phi$, with $A$ a variable, is
satisfied with $A = \emptyset$. If this is combined with a set constructor such
as $\{\cdot \plus \cdot\}$, it is possible to iterate over the elements of the
quantification domain until the `end' is reached: if it is $\emptyset$, then
the quantifier can be eliminated; it if is a variable, then the iteration can
be stopped because we know that we have a good candidate solution for the
quantification domain. Quantification domains are crucial to find out a
decision procedure for formulas where REQ are after RUQ. As we have shown in
Theorem \ref{t:termFUE}, REQ with the same domain variable than a preceding
RUQ, in general, generate infinite feedback loops. These loops can be easily
detected by following the flow of hypothesized elements through quantification
domains. We believe all this is harder to see when the language admits general
quantifiers and is not based on a theory of sets. At the same time, RQ do not
pose a threat on expressiveness when it comes to software verification.
Finally, concerning counterexample generation, within the decidable fragments
presented in this paper, \setlog is not only always able to generate a
counterexample of any given satisfiable formula but it (interactively)
generates a finite representation of all its solutions (Theorem
\ref{equisatisfiable}).

Our work is  closer to CST. CST has been looking for decidability results on
quantified fragments of set theory since at least forty years ago. In many
cases, CST is interested in proving decidability results (in the form of
satisfiability tests) but not so much in providing efficient algorithms or in
implementing them in some software tool. Brevan et al
\cite{https://doi.org/10.1002/cpa.3160340203} present a semi-decision algorithm
for a wide class of quantified formulas where the quantifier-free theory is
decidable. In this work, quantifiers are RQ but no quantified variable can be a
quantified domain of a deeper RQ (see ($\dagger$) below). For some theories the
algorithm becomes complete. In general, the quantifier-free theories are
sub-languages of set theory. In particular they consider a language based on
$\{=,\in\}$. The resulting quantified language allows to express many
set-theoretic operators (e.g. union). From that article, several
researchers of the CST community have found a number of (un)decidability
results about different fragments of quantified languages of set theory
\cite{DBLP:journals/mlq/ParlamentoP92,https://doi.org/10.1002/cpa.3160460104,DBLP:journals/mlq/OmodeoPP96,DBLP:journals/jsyml/BelleP06,DBLP:journals/jsyml/OmodeoP12}.

As can be seen, the decidability results represented by Theorems
\ref{equisatisfiable}-\ref{terminationFEU} have already been proved. On the
contrary, we believe the result of Theorem \ref{t:termFUE} is new. Besides, as
far as we understand, all of our results are new in terms of the algorithm we
use and in particular the set of rewrite rules we use, not to mention the fact
that we put these results to work in a software tool that is able to solve
real-world problems (Section \ref{practice}).

More recently, Cantone and Longo
\cite{DBLP:journals/tcs/CantoneL14,Cantone2011}  worked on the language $\CL$,
part of Cantone's long work on CST. $\CL$ helps to analyze the decidability and
expressiveness of $\LRQBR$\footnote{The same could be achieved by using as a
reference the work on CST by Breban et al.
\cite{https://doi.org/10.1002/cpa.3160340203}. We opted by Cantone's  because
is newer. Nonetheless, Breban's is also duly referenced.}. $\CL$ is a
two-sorted quantified fragment of set theory allowing the following literals:
$x \in y$, $(x,y) \in R$, $x = y$ and $R = S$, where $x$ and $y$ are set
variables and $R$ and $S$ are variables ranging over binary relations. Note
that in $\CL$ sets are \emph{pure} meaning that their elements are sets where
the empty set is the base element (semantics of $\CL$ is given in terms of the
von Neumann standard cumulative hierarchy of sets). Formulas in $\CL$ are
Boolean combinations of expressions of the following two forms: $\forall x_1
\in z_1: \dots \forall x_h \in z_h: \forall (x_{h+1},y_{h+1}) \in R_{h+1}:
\dots \forall (x_n,y_n) \in R_n : \delta$, and the same expression where
$\forall$ is replaced by $\exists$. In these expressions: $\delta$ is a
propositional combination of $\CL$-literals; $x_i, y_i, z_i$ are set variables;
$R_i$ are binary relation variables; and ($\dagger$) no $x_i$ or $y_i$ can also
occur as a $z_j$ (i.e., no quantified variable can occur also as a domain
variable in the same quantifier prefix). Note that in $\CL$ RUQ and REQ cannot
be mixed in the same expression.

$\CL$ is a decidable language which allows to express all the operators  of RA
with the exception of composition. Indeed, $\CL$ only allows to express $R
\circ S \subseteq T$ but the other inclusion cannot be written. The
impossibility to express the other inclusion comes from the fact that RUQ and
REQ cannot be mixed in the same expression, which is tantamount to preserve
decidability of $\CL$. In effect, $T \subseteq R \circ S$ is equivalent to:
\begin{equation}\label{eq:comp}
\forall (x,z) \in T: (\exists (x_1,y_1) \in R: (\exists (y_2,z_1) \in S: x_1 = x \land y_1 = y_2 \land z_1 = z))
\end{equation}
which is not a $\CL$ formula (as it mixes RUQ and REQ).

Now we analyze the decidability and expressiveness of $\LRQBR$ in terms of $\CL$:
\begin{enumerate}
\item $\CL$ sets are not necessarily finite; $\LRQBR$ sets are finite.
However, since we are interested in software verification this is not a real
restriction.
\item $\CL$ is not a parametric language as $\LRQ$, although other works
on CST provide parametric languages in the line of $\CL$
\cite{https://doi.org/10.1002/cpa.3160340203}. Parametrization of $\LRQ$
enables hybrid sets.
\item $\CL$ sets are pure, while $\LRQBR$ sets are hybrid. Pure sets
allow to encode ordered pairs, natural numbers, etc. However, these encodings
tend to reduce the efficiency of solvers. Working directly with hybrid sets
facilitates the integration with efficient solvers for other theories, such as
$\mathcal{LIA}$. For instance, the formula in Example \ref{ex:min} encodes the
minimum of a set which would require a complex $\CL$ formula.
\item $\LRQBR$ extends the decidability result of $\CL$. On one hand, $\CL$
almost expresses RA so if in $\LRQBR$ composition is used as in $\CL$, the
former is a fragment of the latter in what concerns to RA. On the other hand,
$\LRQBR$ allows to fully express composition by a suitable encoding of formula
\eqref{eq:comp}:
\begin{equation*}
\Forall((x,z) \in T, \Exists([(x_1,y_1) \in R, (x_2,z_1) \in S], x_1 = x \land y_1 = y_2 \land z_1 = z))
\end{equation*}
As can be seen, this formula is free of $\fe$ loops as long as the variable of
$T$ is different from the variables of $R$ and $S$. Hence, Cantone and Longo go
to far in restricting $\CL$ as the real problem with composition comes with
formulas such as $R \subseteq R \circ S$ or $S \subseteq R \circ S$. This is
aligned with our results concerning the decidability of
$\mathcal{L}_\mathcal{BR}$ \cite[Section 5.3, Definition
16]{DBLP:journals/jar/CristiaR20}. Finding larger decidable fragments of RA is
important as it is a fundamental theory in Computer Science due to its
expressiveness \cite[last paragraph Section 1]{DBLP:journals/jar/Givant06}.
Besides, as shown in Section \ref{android} with formula \verb+permsDom+,
allowing $\Phi_{\fe}$ formulas is useful in practice.
\item $\LRQBR$ allows quantified variables to occur as domains in the
same quantifier prefix (cf. ($\dagger$) above). Recall, for instance, formula
\verb+notDupPerm+ in Section \ref{android}. These formulas are ruled out from
$\CL$ because they compromise completeness, not soundness \cite[Sections 2 and
4]{https://doi.org/10.1002/cpa.3160340203}. The problem is that there are
formulas not adhering to ($\dagger$) that are satisfied only by infinite sets
when $\in$ is part of the quantifier-free theory \cite{parlamentoPolicriti}. So
if we allow these formulas in \setlog its answers are correct because the tool
is still sound, although it will not terminate for those formulas that are
satisfied only by infinite sets. As shown in Section \ref{android}, allowing
formulas not adhering to ($\dagger$) is useful in practical cases.
\item $\CL$ is proved to be decidable by encoding each of its formulas as
a $\forall_0^\pi$-formula. In turn, $\forall_0^\pi$ is shown to be decidable by
means of the notions of skeletal representation and of its realization
\cite{Cantone2011}. The authors do not provide an algorithm with an obvious
operative semantics for the decidability problem of $\CL$ formulas. Conversely,
by adapting our results on RIS, we provide a simple and concrete solver for
$\LRQ$ (i.e., $\SATRQ$) with CLP properties easily implementable as part of
\setlog. In turn, we provide empirical evidence of its practical capabilities.
The algorithms presented by Breban
\cite{https://doi.org/10.1002/cpa.3160340203} are closer to $\SATRQ$.
\end{enumerate}









\section{\label{concl}Final Remarks}
We have presented a decision procedure for quantifier-free, decidable languages
extended with restricted quantifiers. The decision procedure is based on a
small collection of rewrite rules for primitive set-theoretic operators
($\subseteq, \in, =$). Although all but one of the decidability results
underlying our decision procedure are not new, as far as we understand, the
decision procedure and its rewrite rules are novel. The new decidability result
concerns quantified formulas where a restricted existential quantifier comes
after a restricted universal quantifier. The result is based on building a
graph linking the universally quantified domains with those existentially
quantified. Then, a path analysis is performed to find out whether or not there
is a path from a universally quantified domain to the same domain but making
part of an existential quantifier. Finally, the implementation of the decision
procedure as part of a software tool (\setlog) and its successful application
to real-world, industrial-strength case studies as an automated software
verifier provide empirical evidence of the usefulness of the approach.

Our strongest decidability results are possible by imposing some restrictions
on the quantifier-free language---namely,  that it does not include terms
denoting sets. Although non-trivial languages fulfill these restrictions (e.g.,
linear integer arithmetic), the greatest expressiveness is reached when some of
these restrictions are lifted at the expense of termination. Hence, as a future
work we plan to study what quantifier-free languages preserve termination even
though they support sets to some extent. In particular, the long and
fruitful work on CST should help us in finding those languages.

\bibliographystyle{splncs03}
\bibliography{/home/mcristia/escritos/biblio}

\begin{thebibliography}{10}
\providecommand{\url}[1]{\texttt{#1}}
\providecommand{\urlprefix}{URL }

\bibitem{Abrial:2010:MES:1855020}
Abrial, J.R.: Modeling in Event-B: System and Software Engineering. Cambridge
  University Press, New York, NY, USA, 1st edn. (2010)

\bibitem{DBLP:journals/sttt/AbrialBHHMV10}
Abrial, J., Butler, M.J., Hallerstede, S., Hoang, T.S., Mehta, F., Voisin, L.:
  Rodin: an open toolset for modelling and reasoning in {E}vent-{B}. Int. J.
  Softw. Tools Technol. Transf.  12(6),  447--466 (2010),
  \url{https://doi.org/10.1007/s10009-010-0145-y}

\bibitem{BLP1}
Bell, D.E., LaPadula, L.: Secure computer systems: Mathematical foundations.
  MTR 2547, The MITRE Corporation (May 1973)

\bibitem{BLP2}
Bell, D.E., LaPadula, L.: Secure computer systems: Mathematical model. ESD-TR
  73-278, The MITRE Corporation (Nov 1973)

\bibitem{DBLP:journals/jsyml/BelleP06}
Bell{\`{e}}, D., Parlamento, F.: Truth in {V} for
  {\(\exists\)}\({}^{\mbox{*}}\){\(\forall\)}{\(\forall\)}-sentences is
  decidable. J. Symb. Log.  71(4),  1200--1222 (2006),
  \url{https://doi.org/10.2178/jsl/1164060452}

\bibitem{DBLP:conf/RelMiCS/BerghammerHS14}
Berghammer, R., H{\"{o}}fner, P., Stucke, I.: Automated verification of
  relational while-programs. In: H{\"{o}}fner, P., Jipsen, P., Kahl, W.,
  M{\"{u}}ller, M.E. (eds.) Relational and Algebraic Methods in Computer
  Science - 14th International Conference, RAMiCS 2014, Marienstatt, Germany,
  April 28-May 1, 2014. Proceedings. Lecture Notes in Computer Science, vol.
  8428, pp. 173--190. Springer (2014),
  \url{http://dx.doi.org/10.1007/978-3-319-06251-8\_11}

\bibitem{DBLP:journals/cuza/BetarteCLR16}
Betarte, G., Campo, J.D., Luna, C., Romano, A.: Formal analysis of {A}ndroid's
  permission-based security model,. Sci. Ann. Comp. Sci.  26(1),  27--68
  (2016), \url{https://doi.org/10.7561/SACS.2016.1.27}

\bibitem{DBLP:conf/ictac/BetarteCLR15}
Betarte, G., Campo, J.D., Luna, C.D., Romano, A.: Verifying {A}ndroid's
  permission model. In: Leucker, M., Rueda, C., Valencia, F.D. (eds.)
  Theoretical Aspects of Computing - {ICTAC} 2015 - 12th International
  Colloquium Cali, Colombia, October 29-31, 2015, Proceedings. Lecture Notes in
  Computer Science, vol. 9399, pp. 485--504. Springer (2015),
  \url{https://doi.org/10.1007/978-3-319-25150-9\_28}

\bibitem{DBLP:conf/asm/BoniolW14}
Boniol, F., Wiels, V.: The landing gear system case study. In: Boniol, F.,
  Wiels, V., Ameur, Y.A., Schewe, K. (eds.) {ABZ} 2014: The Landing Gear Case
  Study - Case Study Track, Held at the 4th International Conference on
  Abstract State Machines, Alloy, B, TLA, VDM, and Z, Toulouse, France, June
  2-6, 2014. Proceedings. Communications in Computer and Information Science,
  vol. 433, pp. 1--18. Springer (2014),
  \url{https://doi.org/10.1007/978-3-319-07512-9\_1}

\bibitem{https://doi.org/10.1002/cpa.3160340203}
Breban, M., Ferro, A., Omodeo, E.G., Schwartz, J.T.: Decision procedures for
  elementary sublanguages of set theory. {II}. {F}ormulas involving restricted
  quantifiers, together with ordinal, integer, map, and domain notions.
  Communications on Pure and Applied Mathematics  34(2),  177--195 (1981),
  \url{https://onlinelibrary.wiley.com/doi/abs/10.1002/cpa.3160340203}

\bibitem{10.5555/92143}
Cantone, D., Ferro, A., Omodeo, E.: Computable Set Theory. Clarendon Press, USA
  (1989)

\bibitem{DBLP:journals/tcs/CantoneL14}
Cantone, D., Longo, C.: A decidable two-sorted quantified fragment of set
  theory with ordered pairs and some undecidable extensions. Theor. Comput.
  Sci.  560,  307--325 (2014),
  \url{http://dx.doi.org/10.1016/j.tcs.2014.03.021}

\bibitem{Cantone2011}
Cantone, D., Longo, C., Asmundo, M.N.: A decidable quantified fragment of set
  theory involving ordered pairs with applications to description logics. In:
  Bezem, M. (ed.) Computer Science Logic, 25th International Workshop / 20th
  Annual Conference of the EACSL, {CSL} 2011, September 12-15, 2011, Bergen,
  Norway, Proceedings. LIPIcs, vol.~12, pp. 129--143. Schloss Dagstuhl -
  Leibniz-Zentrum f{\"{u}}r Informatik

\bibitem{DBLP:series/mcs/CantoneOP01}
Cantone, D., Omodeo, E.G., Policriti, A.: Set Theory for Computing - From
  Decision Procedures to Declarative Programming with Sets. Monographs in
  Computer Science, Springer (2001),
  \url{http://dx.doi.org/10.1007/978-1-4757-3452-2}

\bibitem{10.1093/comjnl/bxab030}
Cristi\'a, M., Katz, R.D., Rossi, G.: {Proof Automation in the Theory of Finite
  Sets and Finite Set Relation Algebra}. The Computer Journal  (05 2021),
  \url{https://doi.org/10.1093/comjnl/bxab030}, bxab030

\bibitem{DBLP:journals/jar/CristiaR20}
Cristi{\'{a}}, M., Rossi, G.: Solving quantifier-free first-order constraints
  over finite sets and binary relations. J. Autom. Reason.  64(2),  295--330
  (2020), \url{https://doi.org/10.1007/s10817-019-09520-4}

\bibitem{DBLP:journals/jar/CristiaR21}
Cristi{\'{a}}, M., Rossi, G.: Automated proof of {B}ell-{L}a{P}adula security
  properties. J. Autom. Reason.  65(4),  463--478 (2021),
  \url{https://doi.org/10.1007/s10817-020-09577-6}

\bibitem{DBLP:journals/jar/CristiaR21a}
Cristi{\'{a}}, M., Rossi, G.: Automated reasoning with restricted intensional
  sets. J. Autom. Reason.  65(6),  809--890 (2021),
  \url{https://doi.org/10.1007/s10817-021-09589-w}

\bibitem{DBLP:journals/corr/abs-2112-15147}
Cristi{\'{a}}, M., Rossi, G.: An automatically verified prototype of a landing
  gear system. CoRR  abs/2112.15147 (2021),
  \url{https://arxiv.org/abs/2112.15147}

\bibitem{DBLP:journals/jar/CristiaR21b}
Cristi{\'{a}}, M., Rossi, G.: An automatically verified prototype of the
  {T}okeneer {ID} station specification. J. Autom. Reason.  65(8),  1125--1151
  (2021), \url{https://doi.org/10.1007/s10817-021-09602-2}

\bibitem{DBLP:journals/corr/abs-2105-03005}
Cristi{\'{a}}, M., Rossi, G.: A decision procedure for a theory of finite sets
  with finite integer intervals. CoRR  abs/2105.03005 (2021),
  \url{https://arxiv.org/abs/2105.03005}, under consideration in Theoretical
  Computer Science

\bibitem{cristia_rossi_2021}
Cristi\'a, M., Rossi, G.: Integrating cardinality constraints into constraint
  logic programming with sets. Theory and Practice of Logic Programming pp.
  1--33 (2021), \url{https://doi.org/10.1017/S1471068421000521}

\bibitem{zbMATH07552282}
Cristi{\'a}, M., Rossi, G.: {{\(\{\mathit{log}\}\)}}: set formulas as programs.
  Rend. Ist. Mat. Univ. Trieste  53, ~24 (2021), id/No 23

\bibitem{CristiaRossiSEFM13}
Cristi{\'a}, M., Rossi, G., Frydman, C.S.: $\{log\}$ as a test case generator
  for the {T}est {T}emplate {F}ramework. In: Hierons, R.M., Merayo, M.G.,
  Bravetti, M. (eds.) SEFM. Lecture Notes in Computer Science, vol. 8137, pp.
  229--243. Springer (2013)

\bibitem{DBLP:journals/tplp/CristiaRF15}
Cristi{\'{a}}, M., Rossi, G., Frydman, C.S.: Adding partial functions to
  constraint logic programming with sets. Theory Pract. Log. Program.  15(4-5),
   651--665 (2015), \url{https://doi.org/10.1017/S1471068415000290}

\bibitem{DBLP:conf/types/Luca020}
{De Luca}, G., Luna, C.: Towards a certified reference monitor of the {A}ndroid
  10 permission system. In: de'Liguoro, U., Berardi, S., Altenkirch, T. (eds.)
  26th International Conference on Types for Proofs and Programs, {TYPES} 2020,
  March 2-5, 2020, University of Turin, Italy. LIPIcs, vol. 188, pp. 3:1--3:18.
  Schloss Dagstuhl - Leibniz-Zentrum f{\"{u}}r Informatik (2020),
  \url{https://doi.org/10.4230/LIPIcs.TYPES.2020.3}

\bibitem{Detlefs2005}
Detlefs, D., Nelson, G., Saxe, J.B.: Simplify: a theorem prover for program
  checking  52(3),  365--473

\bibitem{Dovier00}
Dovier, A., Piazza, C., Pontelli, E., Rossi, G.: Sets and constraint logic
  programming. ACM Trans. Program. Lang. Syst.  22(5),  861--931 (2000)

\bibitem{Dovier2006}
Dovier, A., Pontelli, E., Rossi, G.: Set unification. Theory Pract. Log.
  Program.  6(6),  645--701 (2006)

\bibitem{Dutertre01}
Dutertre, B., de~Moura, L.M.: A fast linear-arithmetic solver for {DPLL(T)}.
  In: Ball, T., Jones, R.B. (eds.) CAV. Lecture Notes in Computer Science, vol.
  4144, pp. 81--94. Springer (2006)

\bibitem{DBLP:journals/lmcs/FeldmanPISS19}
Feldman, Y.M.Y., Padon, O., Immerman, N., Sagiv, M., Shoham, S.: Bounded
  quantifier instantiation for checking inductive invariants. Log. Methods
  Comput. Sci.  15(3) (2019), \url{https://doi.org/10.23638/LMCS-15(3:18)2019}

\bibitem{Flanagan2003}
Flanagan, C., Joshi, R., Ou, X., Saxe, J.B.: Theorem proving using lazy proof
  explication. In: Jr., W.A.H., Somenzi, F. (eds.) Computer Aided Verification,
  15th International Conference, {CAV} 2003, Boulder, CO, USA, July 8-12, 2003,
  Proceedings. Lecture Notes in Computer Science, vol. 2725, pp. 355--367.
  Springer, \url{https://doi.org/10.1007/978-3-540-45069-6_34}

\bibitem{Ge2009}
Ge, Y., Barrett, C.W., Tinelli, C.: Solving quantified verification conditions
  using satisfiability modulo theories  55(1-2),  101--122

\bibitem{DBLP:conf/cav/GeM09}
Ge, Y., de~Moura, L.M.: Complete instantiation for quantified formulas in
  satisfiabiliby modulo theories. In: Bouajjani, A., Maler, O. (eds.) Computer
  Aided Verification, 21st International Conference, {CAV} 2009, Grenoble,
  France, June 26 - July 2, 2009. Proceedings. Lecture Notes in Computer
  Science, vol. 5643, pp. 306--320. Springer (2009),
  \url{http://dx.doi.org/10.1007/978-3-642-02658-4_25}

\bibitem{DBLP:journals/jar/Givant06}
Givant, S.: The calculus of relations as a foundation for mathematics. J.
  Autom. Reasoning  37(4),  277--322 (2006),
  \url{http://dx.doi.org/10.1007/s10817-006-9062-x}

\bibitem{DBLP:books/aw/Lamport2002}
Lamport, L.: Specifying Systems, The {TLA+} Language and Tools for Hardware and
  Software Engineers. Addison-Wesley (2002),
  \url{http://research.microsoft.com/users/lamport/tla/book.html}

\bibitem{Leuschel00}
Leuschel, M., Butler, M.: {ProB}: A model checker for {B}. In: Keijiro, A.,
  Gnesi, S., Mandrioli, D. (eds.) FME. Lecture Notes in Computer Science, vol.
  2805, pp. 855--874. Springer-Verlag (2003)

\bibitem{DBLP:journals/cleiej/LunaBCSCG18}
Luna, C., Betarte, G., Campo, J.D., Sanz, C., Cristi{\'{a}}, M., Gorostiaga,
  F.: A formal approach for the verification of the permission-based security
  model of {A}ndroid. {CLEI} Electron. J.  21(2) (2018),
  \url{https://doi.org/10.19153/cleiej.21.2.3}

\bibitem{DBLP:conf/asm/MammarL14}
Mammar, A., Laleau, R.: Modeling a landing gear system in {E}vent-{B}. In:
  Boniol, F., Wiels, V., Ameur, Y.A., Schewe, K. (eds.) {ABZ} 2014: The Landing
  Gear Case Study - Case Study Track, Held at the 4th International Conference
  on Abstract State Machines, Alloy, B, TLA, VDM, and Z, Toulouse, France, June
  2-6, 2014. Proceedings. Communications in Computer and Information Science,
  vol. 433, pp. 80--94. Springer (2014),
  \url{https://doi.org/10.1007/978-3-319-07512-9\_6}

\bibitem{Moura2007}
de~Moura, L.M., Bj{\o}rner, N.: Efficient e-matching for {SMT} solvers. In:
  Pfenning, F. (ed.) Automated Deduction - CADE-21, 21st International
  Conference on Automated Deduction, Bremen, Germany, July 17-20, 2007,
  Proceedings. Lecture Notes in Computer Science, vol. 4603, pp. 183--198.
  Springer, \url{https://doi.org/10.1007/978-3-540-73595-3_13}

\bibitem{DBLP:journals/mlq/OmodeoPP96}
Omodeo, E.G., Parlamento, F., Policriti, A.: Decidability of
  {\(\exists\)}\({}^{\mbox{*}}\){\(\forall\)}-sentences in membership theories.
  Math. Log. Q.  42,  41--58 (1996),
  \url{https://doi.org/10.1002/malq.19960420105}

\bibitem{DBLP:journals/jsyml/OmodeoP12}
Omodeo, E.G., Policriti, A.: The {B}ernays - {S}ch{\"{o}}nfinkel - {R}amsey
  class for set theory: decidability. J. Symb. Log.  77(3),  896--918 (2012),
  \url{https://doi.org/10.2178/jsl/1344862166}

\bibitem{https://doi.org/10.1002/cpa.3160460104}
Parlamento, F., Policriti, A.: Undecidability results for restricted
  universally quantified formulae of set theory. Communications on Pure and
  Applied Mathematics  46(1),  57--73 (1993),
  \url{https://onlinelibrary.wiley.com/doi/abs/10.1002/cpa.3160460104}

\bibitem{parlamentoPolicriti}
Parlamento, F., Policriti, A.: The logically simplest form of the infinity
  axiom  103(1),  274--276

\bibitem{DBLP:journals/mlq/ParlamentoP92}
Parlamento, F., Policriti, A.: The decision problem for restricted universal
  quantification in set theory and the axiom of foundation. Math. Log. Q.
  38(1),  143--156 (1992), \url{https://doi.org/10.1002/malq.19920380110}

\bibitem{setlog}
Rossi, G.: $\{log\}$. \url{http://www.clpset.unipr.it/setlog.Home.html} (2008),
  last access 2022

\bibitem{schneider2001b}
Schneider, S.: The B-method: An Introduction. Cornerstones of computing,
  Palgrave (2001), \url{http://books.google.com.ar/books?id=Krs0OQAACAAJ}

\bibitem{Woodcock00}
Woodcock, J., Davies, J.: Using Z: specification, refinement, and proof.
  Prentice-Hall, Inc., Upper Saddle River, NJ, USA (1996)

\end{thebibliography}

\appendix

\section{\label{app:proofs}Proofs}

\newcommand{\Ppv}{u}
\newcommand{\Fdd}{\Fpv(d)}
\newcommand{\Pdd}{\Ppv(d)}
\newcommand{\nFdd}{\Fpv(d)}
\newcommand{\Fz}{\Fpv(x)}
\newcommand{\Pz}{\Ppv(x)}
\newcommand{\Fd}{\Fpv(d)}
\newcommand{\Pd}{\Ppv(d)}
\renewcommand{\F}{\Fpv}
\renewcommand{\P}{\Ppv}

In the following, a set of the form $\{P(x) : F(x)\}$ (where pattern and filter are
separated by a colon ($:$), instead of a bar ($|$), and the pattern is
\emph{before} the colon) is a shorthand for $\{y : \exists x (P(x) = y \land
F(x))\}$. That is, the set is written in the classic notation for intensional
sets used in mathematics.

\begin{proposition}\label{l:dD}
\begin{gather*}
\begin{split}
\forall d, & D: \\
                & \{x:\{d \plus D\} | \Fpv @ \Ppv\} = \{\Pd : \Fd\}
                \cup \{\Pz : x \in D \land \Fz\}
\end{split}
\end{gather*}
\end{proposition}

\begin{proof}
Taking any $d$ and $D$ we have:
\begin{gather*}
\{x:\{d \plus D\} | \F @ \P\} \\
= \{\Pz : x \in \{d \plus D\} \land \Fz\} \\
= \{\Pz : (x = d \lor x \in D) \land \Fz\} \\
= \{\Pz : (x = d \land \Fz) \lor (x \in D \land \Fz)\} \\
= \{\Pz : x = d \land \Fz\} \cup \{\Pz : x \in D \land \Fz\} \\
= \{\Pd : \Fd\} \cup\{\Pz : x \in D \land \Fz\}
\end{gather*}
\qed
\end{proof}

\begin{lemma}[Equivalence of rule \eqref{forall:iter}]
\begin{gather*}
\forall t, A: t \notin A \implies \\
\quad
\set{t}{A} \cup \risnocp{x:\set{t}{A}}{\Fpv} = \risnocp{x:\set{t}{A}}{\Fpv}
  \iff \Fpv(t) \land A \cup \risnocp{x:A}{\Fpv} = \risnocp{x:A}{\Fpv}
\end{gather*}
\end{lemma}

\begin{proof}
First note that
\begin{equation}\label{eq:disjt}
t \notin A \implies \{t\} \disj A \land \{t\} \disj \risnocp{A}{\Fpv}
\end{equation}
and
\begin{equation}\label{eq:sub1}
\risnocp{x:\{t\}}{\Fpv} \subseteq \{t\}
\end{equation}
and
\begin{equation}\label{eq:teqphit}
\risnocp{x:\{t\}}{\Fpv} = \{t\} \iff \Fpv(t)
\end{equation}

\begin{gather*}
\set{t}{A} \cup \risnocp{x:\set{t}{A}}{\Fpv} = \risnocp{x:\set{t}{A}}{\Fpv} \\
\iff \{t\} \cup A \cup \risnocp{x:\{t\}}{\Fpv} \cup \risnocp{x:A}{\Fpv}
     = \risnocp{x:\{t\}}{\Fpv} \cup \risnocp{x:A}{\Fpv}
       \why{Prop. \ref{l:dD}; semantics $\plus$} \\
\iff \{t\} \cup A \cup \risnocp{x:A}{\Fpv}
     = \risnocp{x:\{t\}}{\Fpv} \cup \risnocp{x:A}{\Fpv}
       \why{\eqref{eq:sub1}; $\{t\}$ in left-hand side} \\
\iff \Fpv(t) \land \{t\} \cup A \cup \risnocp{x:A}{\Fpv}
     = \{t\} \cup \risnocp{x:A}{\Fpv}
       \why{\eqref{eq:teqphit}} \\
\iff \Fpv(t) \land A \cup \risnocp{x:A}{\Fpv}
     = \risnocp{x:A}{\Fpv}
       \why{\eqref{eq:disjt}; basic property of $\disj$ and $\cup$}
\end{gather*}
\end{proof}

\end{document}